\documentclass[a4paper,UKenglish,cleveref, autoref, thm-restate]{lipics-v2021}

\newcommand{\VertexCover}{\textsc{Vertex Cover}}

\newcommand{\TimeCover}{\textsc{MinTimelineCover}}

\newcommand{\ResTimeCover}{\textsc{Restricted Timeline Cover}}

\newcommand{\Untangling}{\textsc{Network Untangling}}

\newcommand{\vof}[1]{\tau({#1})}

\newcommand{\midv}[3]{q_{{#1},{#2},{#3}}}
\renewcommand{\P}{P}

\newcommand{\DigraphCut}{\textsc{Constrained Digraph Pair Cut}}

\newcommand{\TDom}{{\mathcal{T}}}

\newtheorem{problem}{Problem}

\usepackage{tikz}
\usetikzlibrary{decorations.pathreplacing,positioning, arrows.meta,shapes}
\usetikzlibrary{patterns}

\tikzset{main node/.style={circle,draw,minimum size=0.2cm,inner sep=0pt},
            }
\tikzset{example node/.style={circle,draw,minimum size=0.5cm,inner sep=4pt},
}

\newcommand{\ml}[1]{{\color{black}{{#1}}}}

\newcommand{\rd}[1]{{\color{black}{{#1}}}}

\newcommand{\utpos}[1]{u^+_{{#1}}}
\newcommand{\vtpos}[1]{v^+_{{#1}}}
\newcommand{\utneg}[1]{u^-_{{#1}}}
\newcommand{\vtneg}[1]{v^-_{{#1}}}

\title{An FTP Algorithm for Temporal Graph Untangling} 

\author{Riccardo Dondi}{Università degli studi di Bergamo, Italy}{riccardo.dondi@unibg.it}{}{}

\author{Manuel Lafond}{ Universit de Sherbrooke, Canada}{manuel.lafond@usherbrooke.ca}{}{}

\authorrunning{...} 

\Copyright{ } 



\category{} 

\relatedversion{} 

\acknowledgements{...}

\nolinenumbers 

\EventEditors{ }
\EventNoEds{1}
\EventLongTitle{ }
\EventShortTitle{ }
\EventAcronym{ }
\EventYear{ }
\EventDate{ }
\EventLocation{ }
\EventLogo{}
\SeriesVolume{ }
\ArticleNo{ }

\begin{document}

\maketitle

\begin{abstract}
Several classical combinatorial problems have been
considered and analysed on temporal graphs. 
Recently, a variant of \VertexCover{} on temporal
graphs, called \TimeCover{}, has been introduced to summarize
timeline activities in social networks. The problem asks to cover every temporal edge while minimizing the total span of the vertices (where the span of a vertex is the length of the timestamp interval it must remain active in, minus one).
While the problem has been shown to be NP-hard 
even in very restricted cases, its parameterized complexity has not been fully understood.
The problem is known to be in FPT under the span parameter  
only for graphs with two timestamps, but
the parameterized complexity for the general case is open.
We settle this open problem by giving an
FPT algorithm that is based on a combination of iterative
compression and a 
reduction to the \textsc{Digraph Pair Cut}  problem, a powerful problem that has received significant attention recently.
\end{abstract}


\section{Introduction}
\label{sec:introduction}
Temporal graphs are emerging as one of the main
models to describe the dynamics of complex networks.
They describe how relations (edges) change
in in a discrete time domain~\cite{DBLP:journals/jcss/KempeKK02,holme2015modern}, 
while the vertex set is not changing.
The development of algorithms on temporal graphs has mostly focused on finding paths or walks and on
analyzing graph connectivity~\cite{DBLP:journals/jcss/KempeKK02,DBLP:journals/pvldb/WuCHKLX14,DBLP:journals/tkde/WuCKHHW16,DBLP:journals/jcss/Erlebach0K21,DBLP:journals/jcss/ZschocheFMN20,DBLP:journals/tcs/FluschnikMNRZ20,DBLP:conf/iwoca/BumpusM21,DBLP:conf/iwoca/MarinoS21,DBLP:journals/jcss/AkridaMSR21,DBLP:conf/complexnetworks/DondiH21}.
However, several classical problems in computer science
have been recently extended to temporal graphs
and one of the most relevant problem in
graph theory and theoretical computer science, 
\VertexCover{}, has been considered 
in this context~\cite{DBLP:journals/jcss/AkridaMSZ20,DBLP:conf/aaai/HammKMS22,DBLP:journals/datamine/RozenshteinTG21}.

In particular, here we study a variant
of \VertexCover{}, called \Untangling{} introduced in~\cite{DBLP:journals/datamine/RozenshteinTG21}.
\Untangling{} has application in 
discovering event timelines and summarizing temporal networks.
It considers a sequence of temporal
interactions between entities (e.g.  discussions between users in a social network) 
and aims to explain 
the observed interactions with few (and short) \emph{activity intervals} of entities, 
such that each interaction is covered  by at least one of the two entities involved (i.e. at least one of the two entities is active when an interaction between them is observed). 

\Untangling{} can be seen as a variant of \VertexCover{}, where we search for a minimum cover of the interactions, called 
temporal edges.
The size of this temporal vertex cover is based
on the definition of \emph{span} of a vertex,
that is the length of vertex activity.
In particular, the span of a vertex is defined 
as the difference between the maximum and minimum timestamp
where the vertex is active.
Hence, if a vertex is active in exactly one timestamp,
it has a span equal to $0$.

Four combinatorial formulations of \Untangling{} have been defined in \cite{DBLP:journals/datamine/RozenshteinTG21},
varying the definition of vertex activity
(a single interval or $h \geq 2$ intervals)
and the objective function
(minimization of the sum of vertex spans or
minimization of the maximum vertex span).
Here we consider the formulation, denoted by \TimeCover{},
where vertex activity is defined as a single interval and
the objective function is the minimization of the sum
of vertex spans.
Hence, given a temporal graph, \TimeCover{}
searches for a cover of the temporal edges
that has minimum span and
such that each vertex is active in 
one time interval.


The \TimeCover{} problem is known to be NP-hard \cite{DBLP:journals/datamine/RozenshteinTG21}.
The problem is hard also in very restricted cases
when each timestamp contains at most one temporal edge \cite{DBLP:conf/ictcs/Dondi22},
when each vertex has at most two 
incident temporal edges in each timestamp and
the temporal graph is defined over three timestamps 
\cite{DBLP:conf/ictcs/Dondi22}, and
when the temporal graph is defined over two timestamps \cite{DBLP:conf/ijcai/FroeseKZ22}.
Note that, since the span of a vertex activity in 
exactly one timestamp is equal to $0$,
\TimeCover{} is trivially in P when the temporal
graph is defined on a single timestamp,
since in this case any solution of the problem has span $0$.
Furthermore, 
deciding whether there exists a solution
of \TimeCover{} that has span equal to $0$ 
can be decided in polynomial time via a reduction to \textsc{2-SAT}~\cite{DBLP:journals/datamine/RozenshteinTG21}.

\TimeCover{} has been considered also in the parameterized complexity
framework.
The definition of span leads to a problem
where the algorithmic approaches applied to \VertexCover{} cannot be
easily extended for the parameter span of the solution. Indeed, in \VertexCover{}
for each edge we are sure than at least one of
the endpoints must be included in the solution,
thus at least one of the vertex contributes to
the cost of the solution.
This leads to the textbook FPT algorithm of branching over the endpoints of any edge.
For \TimeCover{}, a vertex with span $0$ may cover a temporal
edge, as the vertex can be active
only in the timestamp where the temporal edge
is defined. This makes it more challenging to 
design FPT algorithms when the parameter
is the span of the solution.
In this case, \TimeCover{} is known
to admit a parameterized algorithm 
only when the input temporal graph is defined over
two timestamps  \cite{DBLP:conf/ijcai/FroeseKZ22}, with a
 parameterized reduction to the \textsc{Almost 2-SAT} problem. 
However, the parameterized complexity
of \TimeCover{} for parameter
the span of the solution on general instances
has been left open \cite{DBLP:conf/ijcai/FroeseKZ22,DBLP:conf/ictcs/Dondi22}.
The authors of \cite{DBLP:conf/ijcai/FroeseKZ22}  have also analyzed the parameterized complexity of
the variants of \Untangling{}
proposed in \cite{DBLP:journals/datamine/RozenshteinTG21},
considering other parameters in addition
to the span of the solution:
the number
of vertices of the temporal graph, 
the length of the time domain, and
the number of intervals of vertex
activity. 

\noindent \textbf{Our contributions.}
We solve the open question on 
the parameterized complexity
of \TimeCover{} by showing that the problem is FPT in parameter $k$, the span of a solution, even if the number of timestamps is unbounded. Our algorithm takes time $O^*(2^{5k \log k})$, where the $O^*$ notation hides polynomial factors.
Our algorithm is divided into two phases, each using a different technique.  First, given a temporal graph $G$, we use a variant of iterative compression, where we start from a solution $S$ of span at most $k$ on a subgraph of $G$ induced by a subset of vertices (taken across all timestamps), and then try to maintain such a solution after adding a new vertex of $G$ to the graph under consideration.  This requires us to reorganize which vertices involved in $S$ should be in the solution or not, and in which timestamps.  One challenge is that since the number of such timestamps is unbounded, there are too many ways to decide how to include or not include the vertices that are involved in $S$.  We introduce the notion of a \emph{feasible assignment}, 
which 
allows us to compute how the
vertices in
$S$ can be reorganized (see text for definition).  There are only $2^{O(k \log k)}$ 
ways of reorganizing the vertices in $S$.
We try each such feasible assignments $X$, and we must then find a temporal cover of the whole graph $G$ that ``agrees'' with $X$. 

This leads to the second phase of the algorithm, which decides if such an agreement cover exists through a reduction to a variant of a problem called \textsc{Digraph Pair Cut}.  
In this problem, we receive a directed graph and forbidden pairs of vertices, and we must delete at most $k$ arcs so that a specified source vertex does not reach both vertices from a forbidden pair.  It is known that the problem can be solved in time $O^*(2^k)$.  
In this work, we need a version where the input specifies a set of deletable and undeletable arcs, which we call \DigraphCut{}.
The \textsc{Digraph Pair Cut} problem and its variants have played an important role in devising randomized kernels using matroids~\cite{DBLP:journals/jacm/KratschW20} and, more recently, in establishing a dichotomy in the complexity landscape of constraint satisfaction problems~\cite{kim2022directed,kim2023flow}.
Here, the problem is useful since it can model the implications of choosing or not a vertex in a solution and, in a more challenging way, allows implementing the notion of cost using our definition of span.
We hope that the techniques developed for this reduction can be useful for other variants of temporal graph cover.



\noindent
\ml{\textbf{Overview of the algorithm.} 
Our approach is loosely inspired by some ideas from the FPT algorithm for two timestamps, which is a reduction to \textsc{Almost 2-SAT}~\cite{DBLP:conf/ijcai/FroeseKZ22}. 
 In the latter, one is given a set of clauses with at most two variables and must delete a minimum number of them so that those remaining are satisfiable. 
We do not use \textsc{Almost 2-SAT} directly, but its usage for two timestamps may help understanding the origins of our techniques and the relevance of our reduction to \textsc{Digraph Pair Cut}.

The reduction from \TimeCover{} on two timestamps to \textsc{Almost 2-SAT} associates each vertex $v_i$ with a variable $x(v_i)$, which is true when one should include $v_i$ in a vertex cover and false otherwise; each edge $u_iv_i$ is associated with a clause $x(u_i) \vee x(v_i)$ (here, $v_i$ represents the occurrence of vertex $v$ at timestamp $i \in \{1, 2\}$).  This corresponds to enforcing the inclusion of $u_i$ or $v_i$ in our vertex cover, and we can include enough copies of this clause to make it undeletable. Since our goal is to minimize the number of base vertices $v$ with both $v_1$ and $v_2$ in the cover, we also add a clause $\neg x(v_1) \vee \neg x(v_2)$. Then there is a temporal cover of $G$ of span at most $k$ if and only if one can delete at most $k$ clauses of the latter form to make all remaining clauses satisfiable.  
Even though this reduction produces clauses with only positive or negative clauses, \TimeCover{} does not appear to be much simpler than \textsc{Almost 2-SAT} in terms of FPT algorithms, and studying the SAT formulation seems more approachable.  

For $T \geq 3$ timestamps, the clauses of the form $x(u_i) \vee x(v_i)$ can still be used model the vertex cover requirements, but there seems to be no obvious way to model the span of a cover. One would need to devise a set of clauses of size two such that choosing an interval of $t$ vertices in a cover corresponds to deleting $t - 1$ negative clauses.  
Our idea is to extend current FPT algorithms for \textsc{Almost 2-SAT} to accommodate our cost function.  In \cite{DBLP:journals/jcss/RazgonO09}, the authors propose an iterative compression FPT algorithm that starts from a solution that deletes $k + 1$ clauses, and modifies it into a solution with $k$ clauses, if possible.  The algorithm relies on several clever, but complicated properties of the dependency graph of the clauses (in which vertices are literals and arcs are implications implied by the clauses). This algorithm seems difficult to adapt to our problem. To our knowledge, the only other FPT algorithm for \textsc{Almost 2-SAT} is that of \cite{DBLP:journals/jacm/KratschW20}. 
 This is achieved through a parameterized reduction to \textsc{Digraph Pair Cut}.  At a high level, the idea is to start from an initial guess of assignment for a well-chosen subset of variables, then to construct the dependency graph of the clauses.  A certain chain of implications is enforced by our initial guess, the vertex pairs to separate correspond to contradictory literals, and deleting arcs corresponds to deleting clauses.  It turns out that, with some work, we can skip the \textsc{Almost 2-SAT} formulation and reduce \TimeCover{} to (a variant of) 
\textsc{Directed Pair Cut} 
directly by borrowing some ideas from this reduction.  This is not immediate though.  The first challenge is that the aforementioned ``well-chosen initial guess'' idea cannot be used in our context, and we must develop new tools to enumerate a bounded number of initial guesses from a partial solution (which we call feasible assignment).  The second challenge is that our reduction to our variant of
\textsc{Directed Pair Cut} needs a specific gadget to enforce our cost scheme, while remaining consistent with the idea of modeling the dependency graph of the \textsc{Sat} instance corresponding to the vertex problem at hand.  
}


\section{Preliminaries}
\label{sec:preliminaries}

For an integer $n$, we denote $[n] = \{1,\ldots,n\}$ and
for two integers $i$, $j$,
with $i < j$,
we denote $[i,j] = \{i, i+1, \dots, j-1,j\}$.
Temporal graphs are defined over a discrete time domain $\TDom$, which is a sequence $1, 2 \dots, T$
of timestamps. 
A temporal graph is also defined over a 
set of vertices, called \emph{base vertices}, that
do not change in the time domain
and are defined in all timestamps,
and are associated with \emph{vertices},
which are base vertices defined in specific timestamps.
We use subscripts to denote the timestamp to which a vertex belongs to,
so, for a base vertex
$v$ and $t \in [T]$, we use $v_t$ to denote the occurrence of $v$ in timestamp $t$.
A \emph{temporal edge}
connects two vertices, associated
with distinct base vertices, that belong to the same timestamp.

\begin{definition}
\label{def:TempGraph}
A temporal graph  $G = (V_B,E,\TDom)$ consists of

\begin{enumerate}

\item A time domain $\TDom = \{1,2\ldots,T\}$;

\item A set $V_B$ of \emph{base vertices}; 
$V_B$ has a corresponding set $V(G)$ of \emph{vertices},
which consists of base vertices in specific timestamps, defined as follows:
\[
V(G) = \{ v_t: v \in V_B \wedge t \in [T]\}.
\]


\item A set $E = E(G)$ of temporal edges, which satisfies:
\[
E \subseteq \{u_t v_t : u, v \in V_B, t \in [T] \wedge u \neq v \}.
\]

\end{enumerate}
\end{definition}


For a directed (static) graph $H$, 
we denote by $(u,v)$ an arc from vertex $u$
to vertex $v$ 
(we consider only directed static graphs, not
directed temporal graphs).

Given a temporal graph $G = (V_B, E, \TDom)$ and a set of base vertices
$B \subseteq V_B$, 
we define the set $\vof{B}$ of all vertices of $B$ across all times:
\[
\vof{B} = \{v_t : v \in B \wedge t \in [T]\}.
\]
If $B = \{v\}$, we may write 
$\vof{v}$ instead of $\vof{\{v\}}$. 

Given a set $W \subseteq V(G)$, we denote by
$G[W]$ the subgraph induced by vertices
$W$, i.e. $V(G[W]) = W$ and $E(G[W]) = \{u_t v_t \in E : u_t, v_t \in W\}$.
For a subset 
$W_B \subseteq V_B$ of base vertices, we denote $G[W_B] = G[\vof{W_B}]$.
We also use the notation $G - W_B = G[V_B \setminus W_B]$.  Observe that $G[W_B]$ and $G - W_B$ are temporal graphs over the same time domain as $G$.


In order to define the problem we are 
interested in, we need to define 
the \emph{assignment} of a set of base vertices.

\begin{definition}
\label{def:TempGraph}
Consider a temporal graph  $G = (V_B,E,\TDom)$ and a set $W_B \subseteq V_B$ of base vertices.
An \emph{assignment of $W_B$}
is a subset $X \subseteq \vof{W_B}$ such
that  if $u_p \in X$ and
$u_q \in X$, with $p, q \in [T]$,
then $u_t \in X$, for each $t \in [T]$
with $p \leq t \leq q$.
For a base vertex $u \in W_B$ such that there exists $t \in [T]$ with $u_t \in X$, we denote by 
$\delta(u,X)$, $ \Delta(u,X)$, respectively,
the minimum and maximum timestamp, respectively,
such that $u_{\delta(u,X)}, u_{\Delta(u,X)} \in X$.  If $u_t$ does not exist, then $\delta(u, X) = \Delta(u, X) = 0$.
\end{definition}


If $W_B$ is clear from the context or not relevant, then we may say that $X$ is an assignment, without specifying $W_B$.
Note that, given an assignment $X$ and
a set $\vof{v}$, for some $v \in V_B$,
then 
$
X \cap \vof{v} = \{ v_t : v_t \in X \wedge 
v_t \in \vof{v}\}$ contains vertices for $v$ that belong to a contiguous interval of timestamps.
Consider a set $I \subseteq [T]$ of timestamps.
An assignment $X$ \emph{intersects} $I$ if there exists $v_t \in X$ such that $t \in I$.

Now, we give the definition of \emph{temporal cover}.

\begin{definition}
Given a temporal graph $G=(V_B, \TDom, E)$ 
a \emph{temporal 
cover} of $G$ is an assignment $X$ of $V_B$ such 
that the following properties hold:

\begin{enumerate}

\item For each $v \in V_B$
there exists at least one $v_t \in X$, for some $t \in \TDom$.


\item For each $u_t v_t \in E$, with $t \in [T]$,
at least one of $u_t$, $v_t$ is in $X$.

\end{enumerate}

\end{definition}

For a temporal cover $X$ of $G$, 
the \emph{span} of $v$ in $X$ is defined as:
$sp(v, X) = \Delta(v, X) - \delta(v, X). $ 
Note that if a temporal cover $X$  contains,
for a base vertex $v \in V_B$,
a single vertex $v_t$, then $sp(v,X) = 0$.
The span of $X$, denoted by $sp(X)$, is then defined as:
\[
sp(X) = \sum_{v \in V_B} sp(v, X).
\]
Now, we are able to define \TimeCover{} (an example is presented in 
Fig.~\ref{fig:ex}).

\begin{problem} $($\TimeCover{}$)$ \\ 
\textbf{Input:} A temporal graph $G = (V_B,\TDom,E)$.\\
\textbf{Question:} Does there exist a temporal cover
of $G$ of span at most $k$?
\end{problem}

A temporal cover $S \subseteq V(G)$ of span at most $k$ will sometimes be called a \emph{solution}.
Our goal is to decide whether $\TimeCover{}$ is FPT in parameter $k$.

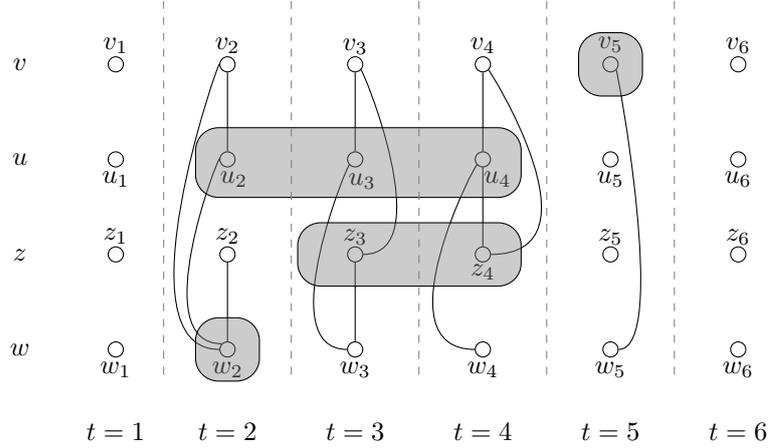
\begin{figure}

  \centering

  \begin{tikzpicture}[scale=0.84]
    \node at (-0.75,0) { $v$};
    \node at (-0.75,-1.5) { $u$};
    \node at (-0.75,-3) { $z$};
    \node at (-0.75,-4.5) { $w$};
               

    \node[main node] (01) at (0.75,0){};
    \node at (0.75,0.3) { $v_1$};

    \node[main node] (02) at (0.75,-1.5) { };
    \node at (0.75,-1.81) { $u_1$};
    
    \node[main node] (03) at (0.75,-3){};
    \node at (0.75,-2.71) {$z_1$};
    
    \node[main node] (04) at (0.75,-4.5){};
    \node at (0.75,-4.82) {$w_1$};


    \node[main node] (1) at (2.5,0){};
    \node at (2.5,0.3) {$v_2$};

    \node[main node] (2) at (2.5,-1.5) { };
    \node at (2.6,-1.81) {$u_2$};
    
    \node[main node] (3) at (2.5,-3){};
    \node at (2.5,-2.71) { $z_2$};
    
    \node[main node] (4) at (2.5,-4.5){};
    \node at (2.5,-4.82) {$w_2$};

    \draw (1.south) -- (2.north);
    

    \draw (1.west) .. controls +(down:-1mm) and +(right:-16mm) ..  (4.west);

     \draw (2.west) .. controls +(down:-1mm) and +(right:-12mm) ..  (4.north west);   

    \draw (3.south) -- (4.north);

    \node[main node] (5) at (4.5,0){};
    \node at (4.5,0.28) {$v_3$};
    
    \node[main node] (6) at (4.5,-1.5) { };
    \node at (4.61,-1.84) {$u_3$};

    \node[main node] (7) at (4.5,-3){};
    \node at (4.5,-2.73) {$z_3$};
    
    \node[main node] (8) at (4.5,-4.5){};
    \node at (4.5,-4.82) {$w_3$};

    \draw (5.south) -- (6.north);

    \draw (5.south east) .. controls +(down:-1mm) and +(right:12mm) .. (7.east);

    \draw (6.south west) .. controls +(down:-1mm) and +(right:-12mm) ..  (8.west);

    \draw (7.south) -- (8.north);
    

    \node[main node] (9) at (6.5,0){};
    \node at (6.5,0.3) {$v_4$};        
    \node[main node] (10) at (6.5,-1.5) { };
    \node at (6.73,-1.81) { $u_4$};
    
    \node[main node] (11) at (6.5,-3){};
    \node at (6.5,-3.27) {$z_4$};
    
    \node[main node] (12) at (6.5,-4.5){};
    \node at (6.5,-4.82) {$w_4$};

    \draw (9.south) -- (10.north);
    
    \draw (10.south) -- (11.north);

    \draw (9.south east) .. controls +(down:-0.3mm) and +(right:+18mm) ..  (11.east);    

    \draw (10.south west) .. controls +(down:-1mm) and +(right:-15mm) ..  (12.west);



    \node[main node] (13) at (8.5,0){};
    \node at (8.5,0.31) {$v_5$};        
    \node[main node] (14) at (8.5,-1.5) { };
    \node at (8.5,-1.81) {$u_5$};
    
    \node[main node] (15) at (8.5,-3){};
    \node at (8.5,-2.73) {$z_5$};
    
    \node[main node] (16) at (8.5,-4.5){};
    \node at (8.5,-4.82) {$w_5$};

    \draw (13.south east) .. controls +(down:-1mm) and +(right:8mm) .. (16.east);


    \node[main node] (1b) at (10.5,0){};
    \node at (10.5,0.3) { $v_6$};        
    \node[main node] (2b) at (10.5,-1.5) { };
    \node at (10.5,-1.81) {$u_6$};
    
    \node[main node] (3b) at (10.5,-3){};
    \node at (10.5,-2.73) {$z_6$};
    
    \node[main node] (4b) at (10.5,-4.5){};
    \node at (10.5,-4.82) {$w_6$};

\draw[fill=gray, fill opacity=0.4, rounded corners=2ex] (8,-0.5) rectangle (9,0.5);

\draw[fill=gray, fill opacity=0.4, rounded corners=2ex] (2,-1) rectangle (7.1,-2.1);

\draw[fill=gray, fill opacity=0.4, rounded corners=2ex] (3.6,-2.5) rectangle (7.1,-3.5);

\draw[fill=gray, fill opacity=0.4, rounded corners=2ex] (2,-4) rectangle (3,-5);

\draw[dashed,thin,gray] (1.5,1)--(1.5,-5);
\draw[dashed,thin,gray] (3.5,1)--(3.5,-5);
\draw[dashed,thin,gray] (5.5,1)--(5.5,-5);
\draw[dashed,thin,gray] (7.5,1)--(7.5,-5);
\draw[dashed,thin,gray] (9.5,1)--(9.5,-5);

\node at (0.75,-5.8){$t=1$};
\node at (2.5,-5.8){$t=2$};
\node at (4.5,-5.8){$t=3$};
\node at (6.5,-5.8){$t=4$};
\node at (8.5,-5.8){$t=5$};
\node at (10.5,-5.8){$t=6$};

\end{tikzpicture}

\caption{An example of \TimeCover{} on a temporal graph 
$G$ consisting of four 
base vertices 
and six timestamps. 
For each timestamp, we draw the temporal edges of $G$, 
for example for $t=2$, the temporal edges
are $v_2 u_2$,
$v_2 w_2$, $u_2 w_2$,
$z_2 w_2$.
Also note that in $t=1$ and $t=6$ no temporal edge is defined.
A temporal cover 
 $X = \{ v_5, u_2, u_3, u_4, z_3, z_4, w_2\}$
 is represented with grey rectangles.
Note that $\delta(v,X) = \Delta(v,X) = 5$,
$\delta(u,X) =2$, $\Delta(u,X) = 4$,
$\delta(z,X) = 3$,  $\Delta(z,X) = 4$,
$\delta(w,X) = \Delta(w,X) = 2$.
It follows that $sp(X) = 3$.
}
\label{fig:ex}
\end{figure}

\section{An FPT Algorithm}
\label{sec:FPT}

In this section we present our FPT algorithm,
which consists of two parts:

\begin{enumerate}

\item The iterative compression technique.

\item A reduction to the \DigraphCut{} problem.

\end{enumerate}

Before presenting the main steps of our algorithm,
we present the main idea and some definitions.
Recall that our parameter, that is the span of a solution of \TimeCover{}, is denoted by $k$.

Consider a temporal graph $G$
and assume we have a temporal cover $S$ of span at most $k$ of the subgraph $G - \{w\}$, for some base vertex $w \in V_B$.
The idea of the iterative compression step
is, starting from $S$, to
show how to decide
in FPT time 
whether there exists a solution of 
\TimeCover{} for $G$. 
This is done by 
solving a subproblem, called
$\ResTimeCover{}$, where we must modify $S$ to consider $w$.
A solution to this subproblem
is computed by branching
on the assignments of base
vertices having a positive span in
$S$ and on $w$, and then reducing
the problem to \DigraphCut{}.
$\ResTimeCover{}$ is defined as follows.

\begin{problem} $($\ResTimeCover{}$)$ \\ 
\textbf{Input:} A temporal graph $G = (V_B,\TDom,E)$, a vertex $w \in V_B$, 
a temporal cover $S$ of
$G - \{w\}$ of span at most
$k$.\\
\textbf{Output:} Does there exist a temporal cover
of $G$ of span at most $k$?
\end{problem}

For technical reasons that will become apparent later, we will assume that the temporal graph contains no edge
at timestamps $1$ and $T$, i.e. $G[\{v_1, v_T : v \in V_B\}]$ is an edgeless graph (as in Fig.~\ref{fig:ex}).  
It is easy to see that if this is not already the case, we can add two such ``dummy'' timestamps, where $G$ does not
contain any temporal edge.
Indeed, since there are no temporal edges
in these two timestamps,
then $G$ has a temporal cover of span at most $k$ if and only the same graph with dummy timestamps has a temporal cover of span at most $k$.  

Informally, if we are able to solve \ResTimeCover{} in FPT time, then we can obtain an FPT algorithm for \TimeCover{} as well.  
Indeed, we can first compute a
temporal cover on a small subset of base vertices
(for example a single vertex), 
and then we can add, one at a time, the other vertices of the graph.  This requires at most $|V_B|$ iterations, and each time a vertex is added, we compute a solution of \ResTimeCover{} to check whether
it is possible to find a temporal
cover of span at most $k$ after the addition of a vertex.

\subsection*{Iterative Compression}

We now present our approach based on iterative
compression to solve
the \ResTimeCover{} problem.
Given a solution $S$ for $G - \{w\}$, we focus on the
vertices of $V_B$ that have
a positive span in $S$ and vertex $w$.
An example of our approach, that
illustrates the sets of base vertices and vertices used by the
algorithm,
is presented in Fig.~\ref{fig:ItCompr1}.

Consider the input of \ResTimeCover{}
that consists of a temporal graph 
$G = (V_B,\TDom,E)$, 
a vertex $w \in V_B$, and
a temporal cover $S$ of $G - \{w\}$ of span at most $k$.
Define the following sets associated with $S$:
\begin{align*}
V_{S} &= \{ v \in V_B : \exists p, q \in [T]
, p < q, \mbox{ such that $v_p, v_q  \in S$ } \} \cup \{  w \} \\
V'_S &= \{ v_t : v_t \in S, v \in V_S \setminus\{ w \} \}
\cup \{ w_t:  t \in [T]
\}.
\end{align*}




The set $V_S$ is defined as the set of base vertices having span greater
than $0$ in $S$, plus vertex $w$. 
$V'_S$ contains the vertices  
in $V(G)$ associated with $V_S$, in particular: 
(1) the vertices 
corresponding to the base vertices in
$V_S \setminus \{ w \}$ that are included in $S$
and (2) vertices corresponding to the base vertex $w$
in every timestamp.

Define the following set $I_S$ of timestamps 
associated with $V_S \setminus \{ w \}$:
\[
I_S = \{t~\in [T]:  u_t \in V'_{S}
\text{
for some $u \in V_S \setminus \{ w \}$  
} \}.
\]

Essentially, $I_S$ contains those timestamps 
where the base vertices of 
$V_S \setminus \{ w \}$, 
that is of span greater than zero,
have associated vertices in $S$. These timestamps are essential
for computing a solution of \ResTimeCover{},
that is to compute whether there exists
a temporal cover of $G[V_B]$ of span
at most $k$ starting from $S$.
We define now the sets of base vertices
and vertices associated
with $S$ and having a span equal to $0$:
\begin{align*}
Z_{S} &=  V_B \setminus V_S  \quad \quad Z'_S =  S \setminus V'_S.
\end{align*}
First, we show two easy properties of $S$ 
and $I_S$ on the
temporal graph $G - \{w\}$.

\begin{lemma}
\label{lem:easybound}
Let $S$ be a solution of \TimeCover{} on instance $G -\{w\}$
and let $I_S$ be the associated set of 
timestamps.
Then $|I_S| \leq 2k$.
\end{lemma}

\begin{lemma}
\label{lem:easypart}
Let $S$ be a solution of \TimeCover{} on instance $G -\{w\}$.
Then, $sp(Z'_S) = 0$.  Moreover,
$Z'_S$
covers each temporal edge of $G -\{w\}$ not covered
by $V'_{S} \setminus \vof{w}$. 
\end{lemma}




Now, we introduce the concept
of feasible assignment, which
is used to ``guess'' how $S$ is rearranged in a solution
of \ResTimeCover{}.
\rd{Recall that an assignment $X$ intersects 
a set $I_S$ of timestamps
if there exists $v_t \in X$ such that $t \in I_S$.}

\begin{definition}[Feasible assignment]
\label{def:feasibleAssignment}
Consider an instance of \ResTimeCover{} 
that consists of
a temporal graph $G = (V_B,\TDom,E)$, 
a vertex $w \in V_B$, 
a temporal cover $S$ of
$G -\{w\}$ of span at 
most $k$, and 
sets $V_S, V'_S$ and $I_S$ associated 
with $S$.
We say that an assignment 
$X \subseteq \vof{V_S}$ of $V_S$ 
is a \emph{feasible assignment} (with respect to $G, S$, and $I_S$) if all of the following conditions hold:

\begin{enumerate}

\item 
the span of $X$ is at most $k$;

\item 
every edge of $G[V_S]$ 
is covered by $X$;

\item 
$X \cap \vof{w}$ is a non-empty assignment of $\{w\}$;


\item 
for every $v \in V_S \setminus \{w\}$, at least one of the following holds: (1) $X \cap \vof{v}$ is empty; (2) $X \cap \vof{v}$ is an assignment of $\{v\}$ that intersects with $I_S$; 
or (3) $X \cap \vof{v}$ contains a vertex $v_t$ such that $v_t w_t \in E$ and $w_t \notin X \cap \vof{w}$.
\end{enumerate}





Given a feasible assignment $X$, we denote 
\[
M_S(X) = \{v \in V_S : X \cap \vof{v} \neq \emptyset\} \quad \quad \quad  N_S(X) = \{v \in V_S : X \cap \vof{v} = \emptyset\}
\]
\end{definition}

\rd{Informally, notice that point 4 considers
the possible cases for a feasible assignment of the vertices of
a base vertex $v \in V_S \setminus \{w\}$ : none of the associated vertices
in $I_S$ belongs to the computed solution (case 4.(1)),
or some of its associated vertices in $I_S$
belongs to the solution, case 4.(2) and case 4.(3),
where the latter case is forced by the need
of covering temporal edge $v_t w_t$,
with $t \in I_s$, not covered by $w_t$. }

Note that $M_S(X)$ and $N_S(X)$ form a partition of $V_S$.
Also note that $G, S$, and $I_S$ are fixed in the remainder, so we assume that all feasible assignments are with respect to $G, S$, and $I_S$ without explicit mention.  
We now relate feasible assignments to temporal covers.

\begin{definition}
    Let $X^*$ be a temporal cover of $G$ and let $X$ be a feasible assignment.  We say that $X^*$ \emph{agrees} with $X$ if:
    \begin{itemize}
        \item 
        for each $v \in M_S(X)$, $X^* \cap \vof{v} = X \cap \vof{v}$;
        
        \item 
        for each $v \in N_S(X)$ and each $t \in I_S$, $X^*$ contains every neighbor $u_t$ of $v_t$ such that $u_t \in \vof{Z_S}$.
        
    \end{itemize}
\end{definition}

The intuition of $X^*$ agreeing with $X$ is as follows.  For $v \in M_S(X)$, $X$ ``knows'' which vertices of $\vof{v}$ should be in the solution, and we require $X^*$ to contain exactly those.  
For $v \in N_S(X)$, we interpret that $X$ does not want any vertex $v_t$ with $t \in I_S$.  Thus, to cover the edges incident to $v_t$ that go outside of $V_S$, we require $X^*$ to contain the other endpoint.  Note an important subtlety: we act ``as if'' $X^*$ should not contain $v_t$ or other vertices of $N_S(X)$ with timestamp in $I_S$, but the definition does not forbid it.  Hence, $X^*$ \emph{can} contain a vertex of $N_S(X)$ in some timestamps of $I_S$, as long as 
$X^*$ contains also its
neighbors (in $I_S$) outside $V_S$.

The main purpose of feasible assignments and agreement is as follows.

\begin{lemma}
\label{lem:CorrectBranch}
    Let $X^*$ be a temporal cover of $G$ of span at most $k$.  Then there exists a feasible assignment $X$ such that $X^*$ agrees with $X$.
\end{lemma}

\begin{proof}
    Construct $X \subseteq X^*$ as follows: 
    add $X^* \cap \vof{w}$ to $X$, and for $v \in V_S \setminus \{w\}$, add $X^* \cap \vof{v}$ to $X$ if and only if $X^* \cap \vof{v}$ intersects with the set $I_S$, or if it contains a vertex $v_t$ incident to an edge $v_tw_t \in E$ such that $w_t \notin X^* \cap \vof{w}$.  Note that since $X^*$ is an assignment of $V_B$, $X$ is an assignment of $V_S$.

    We first focus on arguing that $X$ satisfies each condition of a feasible assignment    (Definition~\ref{def:feasibleAssignment}).  
    For Condition 1, since $X^*$ has span at most $k$ and $X \subseteq X^*$, it is clear that $X$ also has span at most $k$.
    For Condition 3, $X^* \cap \vof{w}$ is non-empty by the definition of a temporal cover, and we added $X^* \cap \vof{w}$ to $X$.
    For Condition 4, we explicitly require in our construction of $X$ that for each $v \in V_S \setminus \{w\}$, if $X \cap \vof{v}$ is non-empty, then it is equal to $X^* \cap \vof{v}$ and it either intersects with $I_S$ or covers an edge not covered by $X \cap \vof{w} = X^* \cap \vof{w}$.

    Let us focus on Condition 2.  Let $u_t v_t \in E(G[V_S])$.  If $u = w$, then if we did not add $w_t$ to $X$, $X^*$ must contain $v_t$ and we added $X^* \cap \vof{v}$ to $X$, thereby covering the edge.  The same holds if $v = w$.  Assume $u \neq w, v \neq w$, and suppose without loss of generality that $X^*$ contains $u_t$ to cover the edge.  
    Suppose for contradiction that $X$ does not cover $u_tv_t$.  Then we did not add $X^* \cap \vof{u}$ to $X$, which implies that $X^* \cap \vof{u}$ does not intersect with $I_S$.  In particular, $t \notin I_S$.  Recall that $S$, the temporal cover of $G - \{w\}$, only intersects with $\vof{u}$ and $\vof{v}$ in timestamps contained in $I_S$.  Hence, $S$ cannot cover $u_tv_t$, a contradiction.  We deduce that $X$ covers every edge.  Therefore, $X$ is a feasible assignment.

    It remains to show that $X^*$ agrees with $X$.  
    For $v \in M_S(X)$, $X^* \cap \vof{v} = X \cap \vof{v}$ by the construction of $X$.
    For $v \in N_S(X)$, there is no $v_t \in X^*$ with $t \in I_S$, as otherwise we would have added $X^* \cap \vof{v}$ to $X$.  For every such $v_t$, $X^*$ must contain all of its neighbors in $\vof{Z_S}$ to cover the edges, as required by the definition of agreement. 
\end{proof}

It remains to show that the number of feasible assignments has bounded size and can be enumerated efficiently.  We first show the latter can be achieved through the following steps.
Start with $X$ as an empty set and
then apply the following steps:

\begin{enumerate}

\item[(1)] Branch into every non-empty assignment $X_w$ of $\{w\}$ of span at most $k$.  In each branch, add the chosen subset $X_w$ to $X$;

\item[(2)] For every edge $v_t w_t \in E(G[V_S])$ such that $w_t \notin X_w$, add $v_t$ to $X$;

\item[(3)] For every 
$v \in V_S \setminus \{ w \}$,
such that $X \cap \vof{v} = \emptyset$ at this moment, branch into $|I_S| + 1$ options: either add no vertex of $\vof{v}$ to $X$, or choose a vertex $v_t$ and add it to $X$, where $t \in I_S$;

\item[(4)] For every $v \in V_S \setminus \{w\}$ such that $X \cap \vof{v} \neq \emptyset$ at this moment, branch into every assignment $X_v$ of $\{v\}$ of span at most $k$ that contains every vertex of $X \cap \vof{v}$ (if no such assignment exists, abort the current branch).  For each such branch, add every vertex of $X_v \setminus X$ to $X$.

\end{enumerate}

\begin{theorem}
\label{teo:enumerateFA}
    The above steps enumerate every feasible assignment in time 
    $O(2^{4 k \log k} T^3 n)$, where $n = |V_B|$.
\end{theorem}

\subsection*{Reducing to \DigraphCut{}}

Our objective is now to list every feasible assignment and, for each of them, to verify whether there is a temporal cover that agrees with it.
More specifically, consider a feasible assignment $X \subseteq \vof{V_S}$. 
Our goal is to decide whether there is a temporal cover $X^*$ of span at most $k$ that
agrees with $X$.  
Since we branch over every possible feasible 
assignment $X$, if there is a temporal cover $X^*$ of $G$ of span at most $k$, then by Theorem~\ref{teo:enumerateFA}  our enumeration will eventually consider an $X$ that $X^*$ agrees with, and hence we will be able to decide of the existence of $X^*$.  

We show that finding $X^*$ reduces to the \DigraphCut{} problem, as we define it below.
For a directed graph $H$, we denote its set of arcs by $A(H)$ (to avoid confusion with $E(G)$, which is used for the edges of an undirected graph $G$).  For $F \subseteq A(H)$, we write $H - F$ for the directed graph with vertex set $V(H)$ and arc set $A(H) \setminus F$.

    \begin{problem} $($\DigraphCut{}$)$ \\ 
\textbf{Input:} A directed graph 
$H = (V(H),A(H))$, a source
vertex $s \in V(H)$, a set
of vertex pairs $P \subseteq {{V(H)} \choose 2}$ called \emph{forbidden pairs}, a subset of arcs $D \subseteq A(H)$ called \emph{deletable arcs}, and an integer $k'$.\\
\textbf{Output:} Does there exist a set of arcs 
$F \subseteq D$
of $H$ such that $|F| \leq k'$ and such that, for each $\{u,v\} \in P$,
at least one of $u$, $v$ is not reachable from $s$ in $H - F$?
\end{problem}

It is known that 
$\DigraphCut{}$ 
can be solved in time $O^*(2^{k'})$~\cite{DBLP:journals/jacm/KratschW20}, but a few remarks are needed before proceeding.  In~\cite{DBLP:journals/jacm/KratschW20}, the authors only provide an algorithm for the \emph{vertex-deletion} variant, and do not consider deletable/undeletable arcs.  It is easy to make an arc undeletable by adding enough parallel paths between the two endpoints, and we show at the end of the section that our formulation of $\DigraphCut{}$ reduces to the simple vertex-deletion variant.
The vertex-deletion variant also admits a randomized polynomial kernel, and other FPT results are known for weighted arc-deletion  variants~\cite{DBLP:conf/stoc/0002KPW22}.

So let us fix a feasible assignment $X$ for the remainder of the section.  We will denote $M_S = M_S(X)$ and $N_S = N_S(X)$.  
We also consider the following set of vertices
associated with $N_S$:
\begin{align*}
N'_S = \{v_2 : v \in N_S\} \quad \quad
\quad
&
N''_S =  \{v_t \in \vof{N_S} : t \in I_S\}.
\end{align*}
For each base vertex $v \in N_S$,
we need $N'_S$ to contain any vertex of $\vof{v}$ in time $[2, T - 1]$, so we choose $v_2$ arbitrarily.  
Then,  
$N''_S$ contains those vertices
$v_t$, with $t \in I_S$,
not chosen by the feasible assignment $X$.
Note that according to our definition of agreement, a solution $X^*$ should contain all the neighbors of $N''_S$ vertices that are in $Z_S$.
Recall that we have defined $Z_S = V_B \setminus V_S$ and 
$Z'_S = S \setminus V'_S$.  
By Lemma~\ref{lem:easypart}
we know that $Z'_S$ covers each temporal edge of $G[V_B \setminus \{w\}]$ not covered by $S \cap V'_S$, and that $sp(Z'_S) = 0$.
We may assume that for each $v \in Z_S$, there is exactly one $t \in [T]$ such that $v_t \in Z'_S$ (there cannot be more than one since $Z'_S$ has span $0$, and if there is no such $t$, we can add any $v_t$ without affecting the span).
Furthermore, we will assume that for each $v \in Z_S$, the vertex $v_t$ in $Z'_S$ is not $v_1$ nor $v_T$. 
 Indeed, since we assume that the first and last timestamps of $G$ have no edges, if $v_t = v_1$ or $v_t = v_T$, then $v_t$ covers no edge and we may safely change $v_p$ to another vertex of $\vof{v}$.

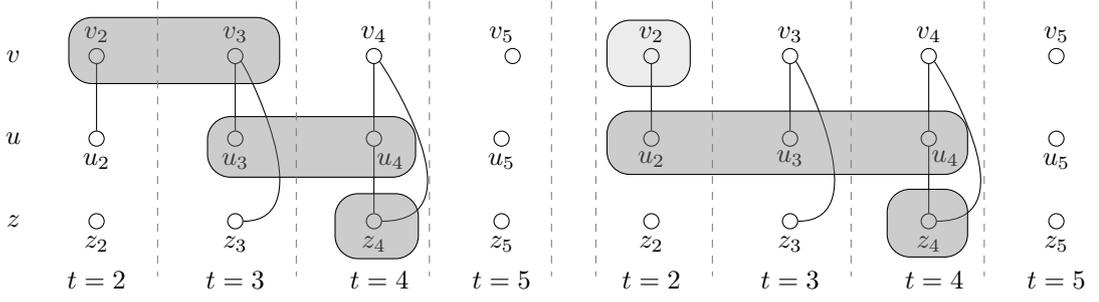
\begin{figure}


  \begin{tikzpicture}[scale=0.73]
    \node at (0,0) {$v$};
    \node at (0,-1.5) {$u$};
    \node at (0,-3) {$z$};


    \node[main node] (1) at (1.5,0){};
    \node at (1.5,0.4) {$v_2$};

    \node[main node] (2) at (1.5,-1.5) { };
    \node at (1.5,-1.9) {$u_2$};
    
    \node[main node] (3) at (1.5,-3){};
    \node at (1.5,-3.42) {$z_2$};

    \draw (1.south) -- (2.north);


    \node[main node] (5) at (4,0){};
    \node at (4,0.4) {$v_3$};
    
    \node[main node] (6) at (4,-1.5) { };
    \node at (4,-1.9) {$u_3$};

    \node[main node] (7) at (4,-3){};
    \node at (4,-3.42) {$z_3$};

    \draw (5.south) -- (6.north);

    \draw (5.south east) .. controls +(down:-1mm) and +(right:15mm) .. (7.east);


    \node[main node] (9) at (6.5,0){};
    \node at (6.5,0.4) { \large{$v_4$}};        
    \node[main node] (10) at (6.5,-1.5) { };
    \node at (6.8,-1.9) {$u_4$};
    
    \node[main node] (11) at (6.5,-3){};
    \node at (6.5,-3.42) {$z_4$};

    \draw (9.south) -- (10.north);
    
    \draw (10.south) -- (11.north);

    \draw (9.south east) .. controls +(down:-0.3mm) and +(right:+18.7mm) ..  (11.east);    



    \node[main node] (13) at (9,0){};
    \node at (8.8,0.4) {$v_5$};        
    \node[main node] (14) at (8.8,-1.5) { };
    \node at (8.8,-1.9) {$u_5$};
    
    \node[main node] (15) at (8.8,-3){};
    \node at (8.8,-3.42) {$z_5$};

\draw[fill=gray, fill opacity=0.4, rounded corners=2ex] (1,-0.5) rectangle (4.8,0.7);

\draw[fill=gray, fill opacity=0.4, rounded corners=2ex] (3.5,-1.1) rectangle (7.25,-2.2);

\draw[fill=gray, fill opacity=0.4, rounded corners=2ex] (5.8,-2.5) rectangle (7.3,-3.68);

\draw[dashed,thin,gray] (2.6,1)--(2.6,-4);
\draw[dashed,thin,gray] (5.1,1)--(5.1,-4);
\draw[dashed,thin,gray] (7.5,1)--(7.5,-4);

\node at (1.5,-4.07) {$t=2$};
\node at (4,-4.07) {$t=3$};
\node at (6.6,-4.07){$t=4$};
\node at (8.8,-4.07) {$t=5$};

\draw[dashed,thin,gray] (9.7,1)--(9.7,-4);
\draw[dashed,thin,gray] (10.5,1)--(10.5,-4);






    \node[main node] (b1) at (11.5,0){};
    \node at (11.5,0.4) {$v_2$};

    \node[main node] (b2) at (11.5,-1.5) { };
    \node at (11.5,-1.87) {$u_2$};
    
    \node[main node] (b3) at (11.5,-3){};
    \node at (11.5,-3.42) {$z_2$};

    \draw (b1.south) -- (b2.north);


    \node[main node] (b5) at (14,0){};
    \node at (14,0.4) {$v_3$};
    
    \node[main node] (b6) at (14,-1.5) { };
    \node at (14,-1.86) {$u_3$};

    \node[main node] (b7) at (14,-3){};
    \node at (14,-3.42) {$z_3$};

    \draw (b5.south) -- (b6.north);

    \draw (b5.south east) .. controls +(down:-1mm) and +(right:15mm) .. (b7.east);


    \node[main node] (b9) at (16.5,0){};
    \node at (16.5,0.4) {$v_4$};        
    \node[main node] (b10) at (16.5,-1.5) { };
    \node at (16.8,-1.86) {$u_4$};
    
    \node[main node] (b11) at (16.5,-3){};
    \node at (16.5,-3.42) {$z_4$};

    \draw (b9.south) -- (b10.north);
    
    \draw (b10.south) -- (b11.north);

    \draw (b9.south east) .. controls +(down:-0.3mm) and +(right:+18mm) ..  (b11.east);    



    \node[main node] (b13) at (18.8,0){};
    \node at (18.8,0.4) {$v_5$};        
    \node[main node] (b14) at (18.8,-1.5) { };
    \node at (18.8,-1.9) {$u_5$};
    
    \node[main node] (b15) at (18.8,-3){};
    \node at (18.8,-3.42) {$z_5$};

\draw[fill=gray, fill opacity=0.15, rounded corners=2ex] (10.7,-0.55) rectangle (12.2,0.65);

\draw[fill=gray, fill opacity=0.4, rounded corners=2ex] (10.7,-1) rectangle (17.2,-2.15);

\draw[fill=gray, fill opacity=0.4, rounded corners=2ex] (15.75,-2.42) rectangle (17.2,-3.65);

\draw[dashed,thin,gray] (12.6,1)--(12.6,-4);
\draw[dashed,thin,gray] (15.1,1)--(15.1,-4);
\draw[dashed,thin,gray] (17.5,1)--(17.5,-4);

\node at (11.5,-4.07) {$t=2$};
\node at (14,-4.07) {$t=3$};
\node at (16.6,-4.07) {$t=4$};
\node at (18.8,-4.07) {$t=5$};

\end{tikzpicture}

\caption{An example of application of iterative compression
(timestamps $1$ and $6$ are not shown as they are edgeless). In the left part, we represent solution $S = \{ v_2, v_3, u_3, u_4, z_4\}$, where the 
vertices in $S$ are highlighted with grey rectangles.
Note that $I_S = \{ 2,3,4 \}$,
$V_S = \{ v,u\}$, 
$V'_S = \{ v_2, v_3, u_3, u_4\}$,
$Z_S =\{ z\}$, $Z'_S = \{ z_4 \}$.
In the right part, we represent in grey
a feasible assignment $X$ associated
with $S$, $X = \{ u_2, u_3, u_4\}$;
in light grey we highlight 
$N'_S = \{ v_2\}$.
The sets associated with $S$ and $X$ are:
$M_S = \{ u\}$, $N_S = \{ v \}$, $N'_S = \{ v_2\}$,
$N''_S =\{ v_2, v_3, v_4 \}$.
The reduction
to \DigraphCut{} eventually leads 
to the solution
of \TimeCover{} represented in Fig.~\ref{fig:ex}.  
}
\label{fig:ItCompr1}
\end{figure}

The following observation will be useful for our reduction to \DigraphCut{}.

\begin{observation}\label{obs:not-disallowed}
Let $u_t v_t \in E(G)$ such that $u \in N_S$ and $v \notin M_S$.  Then $v \in Z_S$ and, if $u_t \notin N''_S$, we have $v_t \in Z'_S$.
\end{observation}


Now, given a feasible assignment $X \subseteq V'_S$, sets $M_S$, $N_S$, $N'_S$, $N''_S$, $Z_S$, and
$Z'_S$,
we present our reduction to the \DigraphCut{} problem.
We construct an instance
of this problem that consists
of the directed graph $H = (V(H), A(H))$, the set of forbidden (unordered) pairs $P \subseteq {{V(H)} \choose 2}$, and the deletable arcs $D \subseteq A(H)$ by applying the following steps. 
The second step in the construction is the most important and is shown in 
Figure~\ref{fig:GeneralTempogadgetsb}.
The intuition of these steps is provided afterwards.

\begin{enumerate}
    \setlength{\itemsep}{0.75em}
   \item \label{step:add-source}
    add to $H$ the source vertex $s$;

    \item \label{step:gadget} 
    for each $v \in Z_S \cup N_S$, 
    let $v_i$ be the vertex of $Z'_S \cup N'_S$, where $i \in [2,T-1]$.
    Add to $H$ the vertices $\vtpos{1}, \ldots, \vtpos{i-1},
    \vtneg{i}, \vtpos{i+1}, \ldots, \vtpos{T}$, the vertices
    $b_{v,j}, c_{v,j}, d_{v,j}$, 
    for $j \in [T] \setminus \{i\}$, 
    and the set of arcs shown in Figure~\ref{fig:GeneralTempogadgetsb},
    that is there are arcs $(\vtpos{j},b_{v,j})$, $(\vtpos{j},c_{v,j})$,
    $(c_{v,j}, d_{v,j})$, $(d_{v,j}, \vtneg{j})$,
    for each $j \in [T] \setminus \{i\}$ 
    and four directed paths: (1) from $b_{v,i-1}$ to $b_{v,1}$, (2) from $c_{v,1}$ to $c_{v,i-1}$,
    (3) from $b_{v,i+1}$ to $b_{v,T}$ and
    (4) from $c_{v,T}$ to $c_{v,i+1}$.

    Add to $D$ the set of deletable arcs
    $(c_{v, j}, d_{v,j})$,
    for $j \in [T] \setminus \{i\}$.

    Then add the following pairs to $P$: 
    \begin{enumerate}
        \item $\{d_{v,h}, b_{v,j}\}$, with 
                $1 \leq h < j \leq i-1$; 

        \item     $\{d_{v,h}, b_{v,j}\}$, with $i+1 \leq j < h \leq T$;

        \item  $\{c_{v,h}, d_{v,j}\}$, with $1 \leq h \leq i-1 \leq i+1 \leq j \leq T$;

        \item  $\{c_{v,h}, d_{v,j}\}$, with 
    $1 \leq j \leq i-1 \leq i+1 \leq h \leq T$.           
    \end{enumerate}
    Note that we have created $T + 3(T - 1) = 4 T-3$ vertices in $H$ in this step.  The subgraph of $H$ induced by these vertices will be called the \emph{gadget corresponding to $v$}.
    
    \item \label{step:edges}
    for each temporal edge $u_t v_t \in E(G)$ such that $u_t, v_t \in \vof{Z_S} \cup (\vof{N_S} \setminus N''_S)$, there are three cases.  
    First note that at least one of $u_t$ or $v_t$ is in $Z'_S$.  Indeed, if $u, v \in Z_S$, this is because an element of $Z'_S$ must cover the temporal edge, and if $u \in N_S$, then $v_t \in Z'_S$ by Observation~\ref{obs:not-disallowed} (or if $v \in N_S, u_t \in Z'_S$).
    The subcases are then:
    \begin{enumerate}
        \item 
        if $u_t, v_t \in Z'_S \cup N'_S$, add the pair $\{ \utneg{t}, \vtneg{t}\}$ to $P$;
        
        \item 
        if $u_t \in Z'_S \cup N'_S, v_t \notin Z'_S \cup N'_S$, add the arc $(\utneg{t}, \vtpos{t})$ to $H$;
        
        \item 
        if $v_t \in Z'_S \cup N'_S, u_t \notin Z'_S \cup N'_S$, add the arc $(\vtneg{t}, \utpos{t})$ to $H$;
        
    \end{enumerate}

    \item \label{step:s-edges}
    for each temporal edge $u_t v_t \in E(G)$ such that $u_t \in (\vof{M_S} \setminus X) \cup N''_S$ and $v_t \in \vof{Z_S}$, there are two cases:
    \begin{enumerate}
        \item 
        if $v_t \notin Z'_S$, add the arc $(s, \vtpos{t})$ to $H$;
        
        \item 
        if $v_t \in Z'_S$, add the pair $\{s, \vtneg{t}\}$ to $P$.
        
    \end{enumerate}

\end{enumerate}

Define $k' = k - sp(X)$.  This concludes the construction.  We will refer to the elements 1, 2, 3, 4 of the above enumeration as the \emph{Steps} of the construction.
Note that the only deletable arcs in $D$ are the arcs $(c_{v,j}, d_{v,j})$ introduced in Step 2.

\begin{figure}

\includegraphics[scale=0.88]{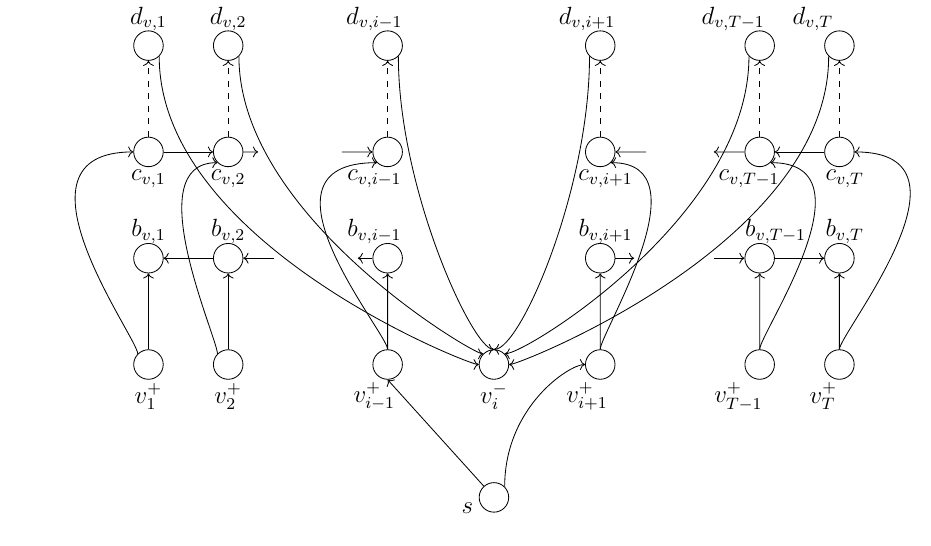}

\caption{Gadget for $v_i \in Z'_S \cup N'_S$,
where $i \in [2,T-1]$. 
\rd{
We assume that there
exist temporal edges $u_t v_t \in E(G)$, where $t \in \{i-1,i+1\}$, such that $u_t \in (\vof{M_S} \setminus X) \cup N''_S$, $v_t \in \vof{Z_S}$ and
$v_t \notin Z'_S$, thus
arcs from $s$ to $\vtpos{t}$
are added.}
The dashed arcs represent deletable arcs.
}
\label{fig:GeneralTempogadgetsb}
\end{figure}

From here, the interpretation of $H$ is that if we delete arc set $F$, then

\begin{enumerate}

\item[(p1)] For $v_t \notin Z'_S \cup N'_S$ we should 
include $v_t$ in $X^*$ if and only if $s$ reaches $\vtpos{t}$ in $H - F$;

\item[(p2)] For $v_t \in Z'_S \cup N'_S$ we should include $v_t$ in $X^*$ if and only if $s$ does \emph{not} reach $\vtneg{t}$ in $H - F$.
 
\end{enumerate}

The idea behind the steps of the construction is then as follows (and is somewhat easier to describe in the reverse order of steps). 
Step 4 describes an initial set of vertices that $s$ is forced to reach, which correspond to vertices that are forced in $X^*$.  A vertex $v_t$ in $\vof{Z_S}$ is forced in $X^*$ if it is in an edge $u_t v_t$ and $u_t \in \vof{M_S}$ but $u_t \notin X$.  By our definition of agreement, $v_t$ is also forced if $u_t \in N_S''$.  Step 4 handles both situations:
if $v_t \notin Z'_S$, we force $s$ to reach $\vtpos{t}$ with the arc $(s, \vtpos{t})$, which is not deletable.  If $v_t \in Z'_S$, then $\vtneg{t} \in V(H)$, and $s$ is forced to \emph{not} reach $\vtneg{t}$ by adding $\{s, \vtneg{t}\}$ to $P$.   By (p1) and (p2), both cases correspond to including $v_t$ in $X^*$.  Then, Step 3 ensures that each temporal edge is ``covered'': for a temporal edge $u_t v_t$, a pair of the form $\{\utneg{t}, \vtneg{t}\}$ in $P$ requires that $s$ does not reach one of the two, i.e. that we include one in $X^*$, and an undeletable arc of the form $(\utneg{t}, \vtpos{t})$ enforces that if $s$ reaches $\utneg{t}$ (i.e. $u_t \notin X^*$), then $s$ reaches $\vtpos{t}$ (i.e. $v_t \in X^*$).  The reason why $Z'_S$ is needed in our construction is that each edge has at least one negative corresponding vertex, so that no other case needs to be considered in Step 3.

Finally, Step~\ref{step:gadget} enforces the number of deleted arcs to correspond to the span of a solution.  That is, it ensures that if we want to add to $X^*$ a set of
$h$ vertices of base vertex $v \in Z_S$ to our solution of 
\ResTimeCover{} (so with a span
equal to $h-1$),
then we have to delete $h-1$ 
deletable arcs of the corresponding gadget of $H$
in order to obtain a solution
to \DigraphCut{} (and vice-versa).
Indeed, consider the gadget in Fig.~\ref{fig:GeneralTempogadgetsb}.
If $v_i$ is not included in $X^*$,
then in the gadget $s$ reaches $h$ positive
vertices $v_l^+, \dots , v_{r}^+$
(and $\vtneg{i}$).
\rd{
It follows that vertices 
$b_{v,l}, \dots, b_{v,r}$,
$c_{v,l}, \dots, c_{v,r}$
and
$d_{v,l}, \dots, d_{v,r}$
are all reachable from $s$.
The pairs $\{ d_{v,x}, b_{v,y} \}$
defined  at Step 2, where 
either $l \leq x \leq y \leq r-1$
if $r < i$,
or $l+1 \leq x \leq y \leq r$ if $l > i$, 
ensures that
arcs
$(c_{v,j}, d_{v,j})$, with $j \in [l,r-1]$
in the former case
or with $j \in [l+1,r]$ 
in the latter case,
are deleted.
}

If $v_i$ is included in $X^*$,
then in the gadget $s$ reaches $h-1$ 
positive vertices $v_l^+, \dots, v_{r}^+$,
with $i \in [l , r]$,
and must not reach negative
vertex $\vtneg{i}$.
It follows that 
vertices 
$b_{v,l}, \dots, b_{v,r}$,
$c_{v,l}, \dots, c_{v,r}$
and
$d_{v,l}, \dots, d_{v,r}$
are all reachable from $s$.
Then
$h-1$ arcs
$(c_{v,j}, d_{v,j})$, with $j \in [l,r] \setminus \{i\},$
must be deleted,
due to the pairs $\{ d_{v,x}, b_{v,y} \}$, $\{ c_{v,x}, d_{v,y} \}$
defined  at Step 2.

Note that Step 2 is the reason we added dummy timestamps $1$ and $T$.  If $v_1$ or $v_T$ were allowed to be in $Z'_S \cup N'_S$, we would need a different gadget for these cases as they behave a bit differently, along with more cases in the proofs.  Adding the edgeless timestamps lets us bypass these cases.
%
%
We now proceed with the details.

\begin{lemma}
\label{lem:RestrictedSol}
There exists a solution of \ResTimeCover{}
that agrees with $X$
if and only if there is $F \subseteq D$ with $|F| \leq k'$ such that $s$ does not reach a forbidden pair in $H - F$.
Moreover, given such a set $F$, a solution of \ResTimeCover{}
can be computed in polynomial time.
\end{lemma}
\begin{proof}[Sketch of the proof]
($\Rightarrow$) Suppose that there exists
a solution $X^*$ of \ResTimeCover{}
that agrees with $X$.
By definition of \ResTimeCover{}, $X^*$ has span at most $k$. 
Note that for $v \in M_S$, the agreement requires that $X^* \cap \vof{v} = X \cap \vof{v}$, and so the span of $v$ in $X^*$ is the same as the span of $v$ in $X$.  Thus 
\[
\sum_{v \in Z_S \cup N_S} sp(v, X^*) \leq k - sp(X) = k'.
\]
We may assume that for every $v \in V_B$, at least one of $v_2, \dots, v_{T-1}$ is in $X^*$, as otherwise we add one arbitrarily without affecting the span (if only $v_1$ or $v_T$ is in $X^*$, remove it first).
For each $v \in Z_S \cup N_S$, consider the gadget corresponding to $v$ in $H$ and delete some of its dashed arcs as follows (we recommend referring to 
Figure~\ref{fig:GeneralTempogadgetsb}).

First, if only one of $\vof{v}$ is in $X^*$, no action is required on the gadget.
So assume that $X^* \cap \vof{v}$ has at least two vertices; in the following
we denote
$v_l = v_{\delta(v,X^*)} $ and
$v_r = v_{\Delta(v,X^*)}$ the vertices
associated with $v$
having minimum and maximum timestamp, respectively, contained in $X^*$.  
We assume that $l,r \in [2, T-1]$ and $l < r$.
Note that $X^* \cap \vof{v} = \{ v_l, v_{l+1}, \dots, v_r\}$.

Let $v_i \in Z'_S \cup N'_S$,
where $i \in [2, T-1]$.  Then
    
\begin{itemize}

    \item  suppose that $l,r \in [2,i-1]$, then:
    delete every arc $(c_{v,q}, d_{v,q})$, 
    with $l \leq q \leq r-1$

    \item  suppose that with $l,r \in [i+1,T-1]$, then:
    delete every arc $(c_{v,q}, d_{v,q})$, 
    with $l+1 \leq q \leq r$

    \item 
    suppose that $l \in [2,i]$ and
    $r \in [i,T-1]$, then: 
    delete every arc $(c_{v,q}, d_{v,q})$, 
    with $l \leq q \leq i-1$, and
    delete every arc $(c_{v,q}, d_{v,q})$, 
    with $i+1 \leq q \leq r$.

\end{itemize}

We see that by construction for all $v \in Z_S \cup N_S$, the number of arcs deleted in the gadget corresponding to $v$ is equal to the 
number of vertices in $X^* \cap \vof{v}$ minus one, that is the
span of $v$ in $X^*$.  
Since these vertices have span at most $k'$, it follows that we deleted at most $k'$ arcs from $H$. 
Denote by $H'$ the graph obtained after deleting the aforementioned arcs.
We argue that in $H'$, $s$ does not reach a forbidden pair.  To this end, we claim the following.

\begin{claim}\label{claim:generals-reach}
For $v \in Z_S \cup N_S$ and $t \in [T]$, 
if $s$ reaches $\vtpos{t}$ in $H'$, then $v_t \in X^*$, 
and if $s$ reaches $\vtneg{t}$ in $H'$, then $v_t \notin X^*$.
\end{claim}

Now, armed with the above claim, we can prove  that in $H'$, $s$ does not reach both vertices of a forbidden pair $q \in P$, thus concluding this direction of the proof.

\medskip

\noindent
($\Leftarrow$)
    Suppose that there is a set 
$F \subseteq D$ with at most $k'$ arcs such that $s$ does not reach a forbidden pair in $H - F$.  
Denote $H' = H - F$.  We construct $X^*$ from $F$, which will also show that it can be reconstructed from $F$ in polynomial time.
Define $X^* \subseteq V(G)$ as follows:

\begin{itemize}
    \item 
    for each $v \in M_S$, add every element of $X \cap \vof{M_S}$ to $X^*$;

     \item 
    for each $v_t \in V(G) \setminus \vof{M_S}$, we add $v_t$ to $X^*$ if and only if one of the following holds:
    (1) $\vtpos{t} \in V(H)$ and $s$ reaches $\vtpos{t}$ in $H'$; or (2) $\vtneg{t} \in V(H)$, and $s$ does \emph{not} reach $\vtneg{t}$ in $H'$;

    \item 
    for each $v_j, v_h \in X^*$ with $j < h$, add $v_t$ to $X^*$ for each $t \in [j+1,h-1]$.
     
\end{itemize} 

Note that $X^*$ agrees with $X$.  Indeed, for $v \in M_S$, there is no gadget corresponding to $v$ in the construction and thus we only add $X \cap \vof{v}$ to $X^*$.  For $u \in N_S$, consider $u_t \in N_S''$ and a neighbor $v_t$ of $u_t$ in $\vof{Z_S}$.
If $v_t \notin Z'_S$, Step~4 adds an undeletable arc from $s$ to $\vtpos{t}$, hence $s$ reaches that vertex and we put $v_t$ in $X^*$.  If $v_t \in Z'_S$, Step~4 adds $\{s, \vtneg{t}\}$ to $P$, and thus $s$ does not reach $\vtneg{t}$ in $H'$, and again we add $v_t$ to $X^*$.  Therefore, we add all the $\vof{Z_S}$ neighbors of $u_t$ to $X^*$, and so it agrees with $X$.
We can prove that $X^*$ covers every temporal edge of $G$ and that $sp(X^*) \leq k$. 
\end{proof}

\subsection*{Wrapping up}

Before concluding, 
we must show that we are able to use the results of~\cite{DBLP:journals/jacm/KratschW20} to get an FPT algorithm for $\DigraphCut{}$, as we have presented it.  
As we mentioned, the FPT algorithm in~\cite{DBLP:journals/jacm/KratschW20} studied the vertex-deletion variant and does not consider undeletable elements, but this is mostly a technicality.
Roughly speaking, in our variant, it suffices to replace each vertex with enough copies of the same vertex, and replace each deletable arc $(u, v)$ with a new vertex, adding arcs from the $u$ copies to that vertex, and arcs from that vertex to the $v$ copies.  Deleting $(u, v)$ corresponds to deleting that new vertex.  For undeletable arcs, we apply the same process but repeat it $k' + 1$ times.

\begin{lemma}
\label{lem:ConstDigraphFPT}
    The $\DigraphCut{}$ problem can be solved in time $O^*(2^{k})$, where $k$ is the number of arcs to delete.
\end{lemma}

We are able now to prove the main result of
our contribution.


\begin{theorem}
\label{teo:final}
\TimeCover{} on a temporal
graph $G=(V_B, E, \TDom)$ can be solved
in time $O^*(2^{5k \log k})$.
\end{theorem}
\begin{proof}
 First, we discuss the correctness of the
 algorithm we presented.
Assume that we have an ordering on 
the base vertices of $G$ and that $v$
is the first vertex of this ordering.
A  solution $S$ of \TimeCover{} on $G[\{v\}]$
is equal to $S = \emptyset$.

Then for $i $, with $i \in [2,|V_B|]$,
let $G_i$ be the temporal graph induced by
the first $i$ vertices and let $w$ be the $i+1$-th
vertex.
Given a solution $S$ of \TimeCover{} on instance 
$G_i$ of span at most $k$, 
we can decide whether
there exists a solution 
of \TimeCover{} on instance 
$G_{i+1}$ by computing whether 
there exists a solution $X^*$ of 
the \ResTimeCover{} problem on instance
$G_i$, $w$, $S$.
By Lemma~\ref{lem:CorrectBranch}
and by Theorem~\ref{teo:enumerateFA}
if there exists such an $X^*$, then 
there exists a feasible assignment $X$
that agrees with $X^*$.
By Lemma~\ref{lem:RestrictedSol} we can compute, via the reduction to \DigraphCut{},
whether there exists a solution of \ResTimeCover{}  on instance on instance
$G_i$, $w$, $S$, and if so obtain such a solution (if no such solution $X^*$ exists, then Lemma~\ref{lem:RestrictedSol} also says that we will never return a solution, since every feasible assignment $X$ that we enumerate will lead to a negative instance of \DigraphCut{}).  Thus the \ResTimeCover{} subproblem is solved correctly, and once it is solved on $G_{|V_B|}$, we have a solution to \TimeCover{}.

Now, we discuss the complexity of the algorithm. We must solve \ResTimeCover{} $|V_B|$ times.  For each iteration, by Theorem~\ref{teo:enumerateFA} 
we can enumerate the feasible assignments
in $O(2^{4k \log k} T^3 n)$ time.
For each such assignment, the
reduction from \ResTimeCover{} to
\DigraphCut{} requires polynomial time, and each generated instance can be solved in time $O^*(2^k)$.
The time dependency on $k$ is thus $O^*(2^{4k \log k} \cdot 2^k)$, which we simplify to $O^*(2^{5k \log k})$.
\end{proof}

\section{Conclusion}
\label{sec:conclusion}

We have presented a randomized FPT algorithm 
for the \TimeCover{} problem,
a variant of \VertexCover{} on temporal
graph recently considered for 
timeline activities summarizations.
We point out some relevant future directions
on this topic: 
(1) to improve, if possible, the time complexity
of \TimeCover{} by obtaining a single exponential time algorithm (of the form $O^*(c^k)$);
(2) to establish whether \TimeCover{} admits a polynomial kernel, possibly randomized (which it might, since \DigraphCut{} famously admits a randomized polynomial kernel); and (3) to extend
the approach to other variants of \Untangling{}.

\bibliography{lipics-v2021-sample-article.bib}

\newpage

\section*{Appendix}

\subsection*{Proof of Lemma~\ref{lem:easybound}}

\setcounter{lemma}{3}

\begin{lemma}
\label{appendxi:lem:easybound}
Let $S$ be a solution of \TimeCover{} on instance $G[V_B \setminus \{w\}]$
and let $I_S$ be the associated set of 
timestamps.
Then $|I_S| \leq 2k$.
\end{lemma}
\begin{proof}
The result follows from the fact that,
since $|V_S \setminus \{ w\}| \leq k$ and $V_S \setminus \{w\}$ contains only vertices of positive span,
$2k$ timestamps contribute at least $k$
to the span of $S$. 
\end{proof}

\subsection*{Proof of Lemma~\ref{lem:easypart}}

\setcounter{lemma}{4}

\begin{lemma}
\label{appendix:lem:easypart}
Let $S$ be a solution of \TimeCover{} on instance $G -\{w\}$.
Then, $sp(Z'_S) = 0$.  Moreover,
$Z'_S$
covers each temporal edge of $G -\{w\}$ not covered
by $V'_{S} \setminus \vof{w}$. 
\end{lemma}
\begin{proof}
Since $S$ is a cover of $G -\{w\}$, it follows
that the temporal edges not covered by $V'_{S} \setminus \{w\}$ must be covered by vertices of 
$Z'_S$. 
Furthermore, by the definition of $Z'_{S}$, 
it follows that $sp(Z'_S)=0$.
\end{proof}

\subsubsection*{Proof of 
Theorem~\ref{teo:enumerateFA}}

\setcounter{theorem}{8}

\begin{theorem}
\label{appendix:teo:enumerateFA}
    The above steps enumerate every feasible assignment in time 
    $O(2^{4 k \log k} T^3 n)$, where $n = |V_B|$.
\end{theorem}

\begin{proof}


We first argue that every feasible assignment is enumerated.  Consider a feasible assignment
$X$.
First, consider the (non-empty) intersection
$X_w = X \cap \vof{w}$. 
Since in Step (1) we branch
into every non-empty assignment
of $\{w\}$, then we eventually enumerate 
$X_w$.  In what follows, we assume that we are in the branch where $X_w$ is added.

Now, consider a vertex 
$v \in V_S \setminus \{ w \}$
and let $X_v = X \cap \vof{v}$.
If no vertex of $\vof{v}$ belongs
to $X$ (hence $X_v$ is empty), then 
Step (3) enumerates
this case (note that Step (2) does not add a vertex of $\vof{v}$ either: $X_v$ empty implies that there is no edge of the form $v_tw_t$ with $w_t \notin X$, since $X$ must cover this edge).
Assume that $X_v$ is non-empty.  By the definition of a feasible assignment, some $v_t$ is in $X_v$ to cover an edge incident to a vertex of $\vof{w}$ not covered by $X_w$, or $v_t$ is in a timestamp of $I_S$.
In the former case, $v_t$ is added in Step (2), and in the latter case, one of the branches of 
Step (3) will add $v_t$ to the set under construction.  
Since the span of $X$ and hence
of $X_v$ is at most $k$,
it follows that $X_v$ is an assignment of $\{v\}$ of span at most $k$, and Step (4) will branch into a case where it adds $X \cap \vof{v}$.  It follows that $X$ will be enumerated at some point.

Now, we discuss the number of
feasible assignments enumerated by Steps (1) -- (4).
Step (1) is computed in $O(T^2)$ time, as there are $O(T^2)$ possible non-empty intervals $X_w$.
For each branch defined at Step (1),
Step (2) can be computed in  
$Tn$ time
as the vertices in $w$ can have at most $Tn$ neighbors.

Step (3), for each vertex $v \in V_S \setminus \{w\}$, branches
into at most $2k + 1$ cases, as 
from Lemma~\ref{lem:easybound} we have that
$|I_S| \leq 2k$. 
Since there are at most $k$ vertices in
$V_S \setminus \{w\}$,
the number of branches explored 
in Step (3) is at most $(2k + 1)^k$. 

As for Step (4), 
for each branch defined in Step (3)
for $v \in V_S \setminus \{w\}$ such that $X \cap \vof{v} \neq \emptyset$,
it branches over $O(k^2)$
possible
choices of timestamps $a$, $b$ 
that are endpoints of $X_v$.  Again since $|V_S| \leq k$, the number of branches explored in Step (4) is at most $O(k^{2k})$.

Thus the overall time to enumerate the feasible assignments with 
Steps (1) -- (4) is
$O(T^2 \cdot Tn \cdot (2k+1)^k k^{2k})$, which is $O(2^{4 k \log k} T^3 n)$ for $k \geq 3$ (and if $k \leq 2$, it is constant and this still holds).
\end{proof}

\subsection*{Proof of Observation~\ref{obs:not-disallowed}}

\setcounter{lemma}{9}

\begin{observation}\label{appendix:obs:not-disallowed}
Let $u_t v_t \in E(G)$ such that $u \in N_S$ and $v \notin M_S$.  Then $v \in Z_S$ and, if $u_t \notin N''_S$, we have $v_t \in Z'_S$.
\end{observation}
\begin{proof}
    There cannot be an edge $u_t v_t$ between two vertices of $\vof{N_S}$, since $X$ must cover every edge in $G[V_S]$ and contains no vertices of $\vof{N_S}$.  Thus $v \notin N_S$, and since $v \notin M_S$, we have $v \in Z_S$.  
    Next suppose that $u_t \notin N''_S$.  Then $t \notin I_S$, implying that $u_t \notin S$.  Hence, the edge must be covered by $v_t$, which is in $S \setminus V'_S = Z'_S$.
\end{proof}

\subsection*{Proof of Lemma~\ref{lem:RestrictedSol}}

\setcounter{lemma}{10}

\begin{lemma}
\label{appendix:lem:RestrictedSol}
There exists a solution of \ResTimeCover{}
that agrees with $X$
if and only if there is $F \subseteq D$ with $|F| \leq k'$ such that $s$ does not reach a forbidden pair in $H - F$.
Moreover, given such a set $F$, a solution of \ResTimeCover{}
can be computed in polynomial time.
\end{lemma}

\begin{proof}
($\Rightarrow$) Suppose that there exists
a solution $X^*$ of \ResTimeCover{}
that agrees with $X$.
By definition of \ResTimeCover{} $X^*$ has span at most $k$. 
Note that for $v \in M_S$, agreement requires that $X^* \cap \vof{v} = X \cap \vof{v}$, and so the span of $v$ in $X^*$ is the same as the span of $v$ in $X$.  Thus 
\[
\sum_{v \in Z_S \cup N_S} sp(v, X^*) \leq k - sp(X) = k'.
\]
We may assume that for every $v \in V_B$, at least one of $v_2, \dots, v_{T-1}$ is in $X^*$, as otherwise we add one arbitrarily without affecting the span (if only $v_1$ or $v_T$ is in $X^*$, remove it first).
For each $v \in Z_S \cup N_S$, consider the gadget corresponding to $v$ in $H$ and delete some of its dashed arcs as follows (we recommend referring to 
Figure~\ref{fig:GeneralTempogadgetsb}).

First, if only one of $\vof{v}$ is in $X^*$, no action is required on the gadget.
So assume that $X^* \cap \vof{v}$ has at least two vertices; in the following
we denote
$v_l = v_{\delta(v,X^*)} $ and
$v_r = v_{\Delta(v,X^*)}$ the vertices
associated with $v$
having minimum and maximum timestamp, respectively, contained in $X^*$.  
We assume that $l,r \in [2, T-1]$ and $l < r$.
Note that $X^* \cap \vof{v} = \{ v_l, v_{l+1}, \dots, v_r\}$.

Let $v_i \in Z'_S \cup N'_S$,
where $i \in [2, T-1]$.  Then
    
\begin{itemize}

    \item  suppose that $l,r \in [2,i-1]$, then:
    delete every arc $(c_{v,q}, d_{v,q})$, 
    with $l \leq q \leq r-1$

    \item  suppose that with $l,r \in [i+1,T-1]$, then:
    delete every arc $(c_{v,q}, d_{v,q})$, 
    with $l+1 \leq q \leq r$

    \item 
    suppose that $l \in [2,i]$ and
    $r \in [i,T-1]$, then: 
    delete every arc $(c_{v,q}, d_{v,q})$, 
    with $l \leq q \leq i-1$, and
    delete every arc $(c_{v,q}, d_{v,q})$, 
    with $i+1 \leq q \leq r$.

\end{itemize}

We see that by construction for all $v \in Z_S \cup N_S$, the number of arcs deleted in the gadget corresponding to $v$ is equal to the 
number of vertices in $X^* \cap \vof{v}$ minus one, that is the
span of $v$ in $X^*$.  
Since these vertices have span at most $k'$, it follows that we deleted at most $k'$ arcs from $H$. 
Denote by $H'$ the graph obtained after deleting the aforementioned arcs.
We argue that in $H'$, $s$ does not reach a forbidden pair.  To this end, we claim the following.

\begin{claim}\label{claim:generals-reach}
For $v \in Z_S \cup N_S$ and $t \in [T]$, 
if $s$ reaches $\vtpos{t}$ in $H'$, then $v_t \in X^*$, 
and if $s$ reaches $\vtneg{t}$ in $H'$, then $v_t \notin X^*$.
\end{claim}
\begin{proof}
The proof is by induction on the distance between $s$ and the vertex. 
As a base case, consider the out-neighbors of $s$ in $H'$.  Suppose that $\vtpos{t}$ is such that $(s, \vtpos{t}) \in A(H')$.
This arc was added to $H$ by Step~\ref{step:s-edges} because $G$ contains a temporal edge $u_t v_t$ with 
$u_t \in (\vof{M_S} \setminus X) \cup N''_S$ and $v_t \in \vof{Z_S}$.  
If $u_t \in \vof{M_S} \setminus X$, then $u_t \notin X$ and $u_t \notin X^*$ either, since it agrees with $X$.  Thus we must have $v_t \in X^*$ to cover the edge.  
If $u_t \in N''_S$, then $v_t \in X^*$ holds by the definition of agreement.
By inspecting the construction of $H$, we see that $s$ does not have an out-neighbor of the form $\vtneg{t}$, and so this suffices for the base case.

Now consider a vertex of the form $\vtpos{t}$ or $\vtneg{t}$ at distance greater than $1$ from $s$ in $H'$, and assume by induction that the claim holds for vertices of this form at a smaller distance.  
Suppose that this vertex is $\vtpos{t}$.  By inspecting the construction, we see that the only possible in-neighbors of $\vtpos{t}$ are either $s$, or some $\utneg{t}$.  The $s$ case was handled as a base case, and so we assume the latter.  
Thus any shortest path from $s$ to $\vtpos{t}$ ends with an arc $(\utneg{t}, \vtpos{t})$.  Since $s$ reaches $\utneg{t}$ with a shorter path, we know by induction that $u_t \notin X^*$.  Moreover, the arc $(\utneg{t}, \vtpos{t})$ must have been created on Step~\ref{step:edges}, and so $u_t v_t \in E(G)$.  Therefore, $v_t \in X^*$ must hold to cover $u_tv_t$, as desired.

So consider instead a vertex of the form $\vtneg{t}$ that $s$ reaches in $H'$.  Assume for contradiction that $v_t \in X^*$. 
By inspecting the construction, we see that the only in-neighbors of $\vtneg{t}$ belong to the gadget corresponding to $v$.
Moreover, the only way to reach $\vtneg{t}$ from $s$ is to go through some $\vtpos{j}$, where $j \in [T] \setminus \{t\}$, and then through some other vertices of the gadget.
Consider a shortest path from $s$ to $\vtneg{t}$ in $H'$, and let $\vtpos{j}$ be the first vertex of the gadget corresponding to $v$ in this path.  By induction, we may assume that $v_j \in X^*$, and hence the span of $v$ is at least one,
since we are currently assuming that $v_t \in X^*$.  

Now, the existence of $\vtneg{t}$ implies that $v_t \in Z'_S \cup N'_S$.
    Consider $\{v_l,v_{l+1}, \dots, v_r\} \subseteq X^*$.  Then we have $l \in [2,t]$, $r \in [t, T - 1]$, 
    and $j \in [l,r] \setminus \{t\}$. In this case, all the arcs
    $(c_{v,y}, d_{v,y})$,
    with $y \in [l,r] \setminus \{t\}$ are removed.  Moreover, by inspecting the steps of the construction, we see that the only out-neighbors of $\vtpos{j}$ are $b_{v,j}$ and $c_{v, j}$. 
    Thus $\vtpos{j}$ can only reach $\vtneg{t}$ through a $c_{v, y}$ and then a $d_{v, y}$ vertex.  
    However, because of our deletions, $v_j^+$ cannot reach any 
    vertex $d_{v,y}$, thus it cannot
    reach $\vtneg{t}$, leading to a contradiction.
We deduce that $v_t \notin X^*$,
which concludes the proof of the claim.
\end{proof}

Now, armed with the above claim, assume for contradiction that in $H'$, $s$ reaches both vertices of a forbidden pair $q \in P$.
If $q$ was created on Step~\ref{step:edges}, then $q = \{\utneg{t}, \vtneg{t}\}$, where $u_t v_t$ is an edge of $G$.  
By Claim~\ref{claim:generals-reach}, this implies that $u_t, v_t \notin X^*$, a contradiction since $X^*$ would not cover the edge.
If $q$ was created on Step~\ref{step:s-edges}, then $q = \{s, \vtneg{t}\}$, where there is an edge $u_t v_t \in E(G)$ such that 
$u_t \in (\vof{M_S} \setminus X) \cup N''_S$ and $v_t \in \vof{Z_S}$.
Under the assumption that $s$ reaches $\vtneg{t}$, Claim~\ref{claim:generals-reach} implies that $v_t \notin X^*$. 
If $u_t \in N''_S$, this contradicts the fact that $X^*$ agrees with $X$, since $v_t$ should be in $X^*$.  
If $u_t \in \vof{M_S} \setminus X$, then $u_t \notin X$ implies that $u_t \notin X^*$ , reaching a contradiction since the edge $u_t v_t$ is not covered.

We may thus assume that $q$ was created on Step~\ref{step:gadget}.  Let $v \in Z_S \cup N_S$ be the base vertex for which the corresponding gadget contains the two vertices of $q$.

Let $v_i \in Z'_S \cup N'_S$, where 
$i \in [2, T - 1]$. 
Assume that $q$ is 
$\{d_{v,h}, b_{v,j}\}$ with $h < j < i$.
Then $s$ must reach $\vtpos{l}$,
$c_{v,l}$, $c_{v,h}$ (possibly identical
to $c_{v,l}$) and $d_{v,h}$, with $l \leq h$.
On the other hand, since $s$ reaches $b_{v,j}$, we have that $s$ must reach 
$\vtpos{r}$ and $b_{v, r}$, where $r \in [j, i - 1]$.
 By Claim~\ref{claim:generals-reach}, $v_l, v_r \in X^*$, in which case we deleted the deletable arc 
$(c_{v,h}, d_{v,h})$ (since $h < j$, and either the first or third situation arises in our list of three cases of arc deletions).
Thus $s$ cannot reach $d_{v,h}$.

Likewise, assume that $q$ is 
$\{d_{v,h}, b_{v,j}\}$ with $i < j < h$.
Then $s$ must reach 
$c_{v,r}$, $c_{v,h}$ (possibly identical
to $c_{v,r}$) and $d_{v,h}$, with $r \geq h$.
On the other hand, $s$ must reach 
$\vtpos{l}$, with $l \leq j <h$.
 By Claim~\ref{claim:generals-reach}, $v_l, v_r \in X^*$, in which case we deleted the deletable arc 
$(c_{v,h}, d_{v,h})$.
Thus $s$ cannot reach $d_{v,h}$.

Assume that $q$ is 
$\{c_{v,h}, d_{v,j}\}$ 
or $\{c_{v,j}, d_{v,h}\}$
with $h < i < j $.
Then $s$ must reach 
$\vtpos{l}$,
$c_{v,l}$, $c_{v,h}$, possibly identical
to $c_{v,l}$ 
and $d_{v,h}$, if $(c_{v,h}, d_{v,h})$ is not deleted.
On the other hand, $s$ must reach $c_{v,r}$, $c_{v, j}$, possibly identical
to $c_{v,r}$, and $d_{v,j}$,
if $(c_{v,i}, d_{v,j})$ is not deleted.
 By Claim~\ref{claim:generals-reach}, $v_l, v_r \in X^*$, in which case we deleted the deletable arcs 
$(c_{v,h}, d_{v,h})$ and $(c_{v,j}, d_{v,j})$.
Thus $s$ cannot reach $d_{v,h}$ and
$d_{v,j}$.

Having handled every forbidden pair, we deduce that we can remove at most $k'$ edges from $H$ so that $s$ does not reach any of them.

\medskip

\noindent
($\Leftarrow$)
    Suppose that there is a set 
$F \subseteq D$ with at most $k'$ arcs such that $s$ does not reach a forbidden pair in $H - F$.  
Denote $H' = H - F$.  We construct $X^*$ from $F$, which will also show that it can be reconstructed from $F$ in polynomial time.
Define $X^* \subseteq V(G)$ as follows:

\begin{itemize}
    \item 
    for each $v \in M_S$, add every element of $X \cap \vof{M_S}$ to $X^*$;

     \item 
    for each $v_t \in V(G) \setminus \vof{M_S}$, we add $v_t$ to $X^*$ if and only if one of the following holds:
    (1) $\vtpos{t} \in V(H)$ and $s$ reaches $\vtpos{t}$ in $H'$; or (2) $\vtneg{t} \in V(H)$, and $s$ does \emph{not} reach $\vtneg{t}$ in $H'$;

    \item 
    for each $v_j, v_h \in X^*$ with $j < h$, add $v_t$ to $X^*$ for each $t \in [j+1,h-1]$.
     
\end{itemize} 

Note that $X^*$ agrees with $X$.  Indeed, for $v \in M_S$, there is no gadget corresponding to $v$ in the construction and thus we only add $X \cap \vof{v}$ to $X^*$.  For $u \in N_S$, consider $u_t \in N_S''$ and a neighbor $v_t$ of $u_t$ in $\vof{Z_S}$.
If $v_t \notin Z'_S$, Step~4 adds an undeletable arc from $s$ to $\vtpos{t}$, hence $s$ reaches that vertex and we put $v_t$ in $X^*$.  If $v_t \in Z'_S$, Step~4 adds $\{s, \vtneg{t}\}$ to $P$, and thus $s$ does not reach $\vtneg{t}$ in $H'$, and again we add $v_t$ to $X^*$.  Therefore, we add all the $\vof{Z_S}$ neighbors of $u_t$ to $X^*$, and so it agrees with $X$.

We claim that $X^*$ covers every temporal edge of $G$.  Since $X$ is a feasible assignment, every temporal edge $u_t v_t \in E(G)$ with $u, v \in V_S$ is covered by $X$, and thus also by $X^*$ since $X \subseteq X^*$.

Next consider a temporal edge $u_t v_t \in E(G)$ with $u_t \in \vof{M_S} \cup N''_S$ and $v_t \in \vof{Z_S}$.  
If $u_t \in X$, the temporal edge is covered since $X \subseteq X^*$.  So assume that $u_t \notin X$, i.e. $u_t \in (\vof{M_S} \setminus X) \cup N''_S$.  If $v_t \notin Z'_S$, Step~\ref{step:s-edges} adds an undeletable arc from $s$ to $\vtpos{t}$, which implies that this arc is in $H'$.  Thus $s$ reaches $\vtpos{t}$ and $v_t \in X^*$ by construction, and $u_t v_t$ is covered.  If $v_t \in Z'_S$, then $\{s, \vtneg{t}\}$ is in $P$ owing to Step~\ref{step:s-edges}, and thus $s$ does not reach $\vtneg{t}$ in $H'$.  Again by construction, $v_t \in X^*$.

Finally, consider a temporal edge $u_t v_t \in E(G)$ with $u_t, v_t \in \vof{Z_S} \cup (\vof{N_S} \setminus N''_S$).
As argued in Step~3, we know that at least one of $u_t$ or $v_t$ is in $Z'_S$.  Suppose that both $u_t, v_t \in Z'_S \cup N'_S$.  
Then by Step~\ref{step:edges}, $\{\utneg{t}, \vtneg{t}\} \in P$, and there is at least one of the two that $s$ does not reach.  By construction, one of $u_t$ or $v_t$ is in $X^*$ and the temporal edge is covered. 
So suppose, without loss of generality, that $u_t \in Z'_S \cup N'_S, v_t \notin Z'_S \cup N'_S$.
Then $\utneg{t} \in V(H)$ and $\vtpos{t} \in V(H)$.
If $s$ does not reach $\utneg{t}$ in $H'$, then $u_t \in X^*$ and the temporal edge is covered. Thus we may assume that $s$ reaches $\utneg{t}$ in $H'$.  By Step~\ref{step:edges}, there is an undeletable arc $(\utneg{t}, \vtpos{t})$  in $H$, and thus in $H'$, which implies that $s$ reaches $\vtpos{t}$.  Thus $v_t \in X^*$ and the temporal edge is covered.
Note that we have covered every case of a possible temporal edge, and we deduce that $X^*$ covers every temporal edge.

We next claim that $sp(X^*) \leq k$. 
Since $X^* \cap \vof{M_S} = X$, the  vertices in $M_S$ have span equal
to $sp(X)$.
We must argue that the vertices of $V_B \setminus M_S$ have a span of at most $k' = k - sp(X)$.
Consider a vertex $v \in V_B \setminus M_S = N_S \cup Z_S$ that has span $sp(v, X^*)$ 
more than $0$ in $X^*$
(recall that 
$sp(v, X^*)$  denotes the span of $v$ in $X^*$). 
We want to show that $sp(v, X^*)$
edges of $H$ were deleted in the gadget corresponding to $v$.

In the following we denote by $v_l$
and $v_r$, with $l,r \in [2, T - 1]$
the minimum  and maximum timestamp,
respectively, such that $v_l \in X^*$
and $v_r \in X^*$.

Let $v_i \in Z'_S \cup N'_S$, where $i \in [l, r]$.
Suppose that $r < i$.
Then by the construction of $X^*$, $s$ 
reaches $\vtpos{l}$ and $\vtpos{r}$,
hence $s$ reaches 
(1) $c_{v,l}$ and thus
$c_{v,j}$,
for each $j \in [l,i-1]$,
and (2) $b_{v,r}$ and  thus
$b_{v,j}$, for each $j \in [1,r]$. 
Thus the arcs $(c_{v,j}, d_{v,j})$,
with $j \in  [l,r-1]$, 
have to be deleted due the forbidden pairs 
$\{d_{v,j}, b_{v,y} \}$, with $l \leq j < y \leq r$.  This amounts to $r -l $ deletions, which is the span of $v$ in $X^*$.

Suppose instead that
with $i < l$.
Similarly to the previous case, by the construction, 
$s$  reaches $\vtpos{l}$ and $\vtpos{r}$,
hence $s$ reaches
(1) $c_{v,r}$ and thus
$c_{v,j}$,
for each $j \in [i+1,r]$, 
and (2) $b_{v,l}$ and  thus
$b_{v,j}$, for each $j$ with 
$j \in [l, T]$.
Thus the arcs $(c_{v,j}, d_{v,j})$,
for each $j \in [l+1,r]$,
have to be deleted due the forbidden pairs 
$\{d_{v,j}, b_{v,y}\}$, with $l \leq y < j \leq r$.  Again, this amounts to $r - l $ deletions, which is the span of $v$
in $X^*$.

Finally, suppose that $l \leq i \leq r$.
We have three cases depending on the
fact that $l = i$, $r = i$ or $l < i < r$.
Consider the first case
$l = i < r$.  Thus by the construction of $X^*$, $s$ reaches $v_r^+$ but does not reach $v_i^-$.
Moreover, $s$ reaches $c_{v,r}$
thus
$s$ reaches $c_{v,j}$,
for each $j \in  [i+1,r]$.
Thus arcs $(c_{v,j}, d_{v,j})$,
for each $j \in [i+1,r]$,
have to be deleted in order to make $\vtneg{i}$
not reachable from $s$. 
This amounts to 
$r- i = r- l$ deletions,
which is the span of $v$
in $X^*$.

Consider the second case
$l < i = r$. 
Similarly to the previous case, 
$s$ reaches $\vtpos{l}$ but does not reach $\vtneg{i}$.  Moreover, $s$ reaches $c_{v,j}$,
for each $j$ with $j \in [l,i-1]$.
Thus arcs $(c_{v,j}, d_{v,j})$,
for each $j \in [l,i-1]$
have to be deleted in order to make $\vtneg{i}$
not reachable from $s$.  This amounts to $i - j = r - l$ deletions,
which is the span of $v$
in $X^*$.

Finally, consider the third case
$l < i < r$.
Then arcs 
$(c_{v,j}, d_{v,j})$, with 
$j \in [l,i-1]$
and $(c_{v,j}, d_{v,j})$, with 
$j  \in [i+1,r]$,
have to be deleted due to forbidden
pairs
$\{c_{v,j}, d_{v,z}\}$,
with $j < i < z$ and
forbidden pairs 
$\{c_{v,z}, d_{v,j}\}$,
with $z < i < j$.
This requires $i - l + r - i = r - l$ deletions,  which is the span of $v$
in $X^*$.

We thus see that each vertex $v$ of 
$V_B \setminus M_S$ has a span that is at most the number of arcs deleted in the gadget of $H$ corresponding to $v$.  
Therefore, $X^*$ is a temporal cover of span at most $sp(X) + k' \leq k$, thus completing
the proof.
\end{proof}

\subsection*{Proof of 
Lemma~\ref{lem:ConstDigraphFPT}}

\setcounter{lemma}{12}

\begin{lemma}
    The $\DigraphCut{}$ problem can be solved in time $O^*(2^{k})$, where $k$ is the number of arcs to delete.
\end{lemma}



\begin{proof}
We call \textsc{Vertex-Deletion Digraph Pair Cut} the problem in which, given a directed graph $H$, a source $s \in V(H)$, pairs $\P \subseteq {V \choose 2}$, and integer $k$, we must decide whether there is $R \subseteq V(H) \setminus \{s\}$ with $|R| \leq k$ such that in $H - R$, $s$ does not reach both $u$ and $v$ for every $\{u, v\} \in \P$ (note that $H - R$ removes vertices here, not arcs).  
In~\cite[Theorem 6.1]{DBLP:journals/jacm/KratschW20}, this problem was shown to be solvable in time $O^*(2^k)$.
We show that $\DigraphCut{}$ as we defined it reduces to \textsc{Vertex-Deletion Digraph Pair Cut}, with the same parameter value $k$.

Suppose that we have an instance of \textsc{Digraph Pair Cut}, with directed graph $H$, source $s$, pairs $\P$, deletable arcs $D$, and integer $k$.  
From this instance, obtain an instance of \textsc{Vertex-Deletion Digraph Pair Cut} with directed graph $H'$, source $s'$, pairs $\P'$, and integer $k$ as follows (note that $k$ is unchanged).  First for each $u \in V(H)$, add $k + 1$ copies $u^1, \ldots, u^{k+1}$ of $u$ to $V(H')$.  Also add a new vertex $s'$ to $V(H')$, which serves as the source for the modified instance. 
Add to $A(H')$ the set of arcs $(s', s^1), \ldots, (s', s^{k+1})$.
Then for each deletable arc $(u, v) \in D$, add to $H'$ a new vertex $\midv{u}{v}{1}$, and the set of arcs 
\[
\{(u^i, \midv{u}{v}{1}) : i \in [k+1]\} \cup \{(\midv{u}{v}{1}, v^j : j \in [k+1]\}
\]
Finally, for each undeletable arc $(u, v) \in A(H) \setminus D$, add to $H'$ the $k + 1$ new vertices $\midv{u}{v}{1}, \midv{u}{v}{2}, \ldots, \midv{u}{v}{k+1}$, and then 
 add to $A(H')$ the set of arcs 
\[
\{(u^i, \midv{u}{v}{l}) : i \in [k+1], l \in [k+1]\} \cup \{(\midv{u}{v}{l}, v^j) : l \in [k+1], j \in [k+1]\}
\]
In other words, each vertex of $H$ has $k + 1$ corresponding vertices in $H'$, making the latter pointless to delete.  For $(u, v) \in D$, deleting the arc corresponds to deleting $\midv{u}{v}{1}$ since it removes the path of length $2$ from every $u^i$ to every $v^j$.  For $(u, v) \in A(H) \setminus D$, there are too many $\midv{u}{v}{l}$ copies, making them pointless to delete.

Finally, for each $\{u, v\} \in \P$, we add to $\P'$ all the pairs $\{u^i, v^j\}$ for every $i, j \in [k+1]$.

Assume that there is $F \subseteq D$ with $|F| \leq k$ such that $s$ reaches no pair of $P$ in $H - F$.  In $H'$, we delete the set of vertices $R = \{\midv{u}{v}{1} : (u, v) \in F \}$.
Note that $|R| = |F| \leq k$.
Suppose for contradiction that in $H' - R$, $s'$ reaches both $\{u^i, v^j\} \in \P'$.  Since every arc of $H'$ is incident to a vertex of the form $\midv{u}{v}{l}$, in $H'$, the path from $s'$ to $u^i$ in $H' - R$ has the form 
\[s' \rightarrow s^{b_0} \rightarrow \midv{s}{x_1}{a_1} \rightarrow x_1^{b_1} \rightarrow \midv{x_1}{x_2}{a_2} \rightarrow x_2^{b_2} \rightarrow \midv{x_2}{x_3}{a_3} \rightarrow \ldots \rightarrow \midv{x_l}{u}{a_{l+1}} \rightarrow u^i
\] 
for some vertices $x_1, \ldots, x_l \in V(H)$ and indices $b_0, a_1, b_1, \ldots, a_{l+1}$.
By our construction of $R$, this means that in $H - F$, all the arcs $(s, x_1), (x_1, x_2), \ldots, (x_l, u)$ are present, and that $s$ reaches $u$ in $H - F$.  By the same logic, $s$ also reaches $v$ in $H - F$, a contradiction since $\{u^i, v^j\} \in \P'$ implies that $\{u, v\} \in \P$.  Thus $R$ is a solution for $H'$.

Conversely, assume that there is $R \subseteq V(H') \setminus \{s'\}$ with $|R| \leq k$ such that $s$ reaches no pair of $P'$ in $H' - R$.  
We may assume that $R$ does not contain a vertex $u^i$ with $u \in V(H)$, since $R$ cannot contain every copy of $u$.  Likewise, for undeletable $(u, v) \in A(H) \setminus D$, we may assume that $R$ does not contain a vertex $\midv{u}{v}{l}$ since $R$ cannot contain every copy.  Therefore, we may assume that $R$ only contains vertices of the form $\midv{u}{v}{1}$, where $(u, v) \in D$.
Define $F = \{(u, v) : \midv{u}{v}{1} \in R \}$.  
Note that $F$ has at most $|R| \leq k$ arcs, and they are all deletable.
Now suppose for contradiction that $s$ reaches both $u, v$ for $\{u,v\} \in \P$ in $H - F$.
Hence in $H - F$ there are paths $P_u, P_v$ from $s$ to $u$ and $v$, respectively. 
Note that for every arc $(x, y)$ of $P_u$, the vertex $\midv{x}{y}{1}$ is still present in $H' - R$.  Thus by replacing every $(x, y)$ with the subpath $x^1, \midv{x}{y}{1}, y^1$, we can obtain a path form $s'$ to $u^1$ in $H' - R$.  Likewise, there is a path from $s'$ to $v^1$ in $H' - R$.
This is a contradiction since $\{u, v\} \in \P$ implies that $\{u^i, v^j\} \in \P'$.  Hence $F$ is a valid solution for $H$.

To conclude, constructing $H'$ can clearly be done in time polynomial in $|V(H)| + |A(H)|$.  It follows that we can solve the instance $H$ in time $O^*(2^k)$ by constructing $H'$ and solving the vertex-deletion variant on it. 
\end{proof}

\end{document}